\documentclass{article}
\usepackage{graphicx} 
\usepackage[utf8]{inputenc} 
\usepackage[T1]{fontenc}
\usepackage{url}
\usepackage{ifthen}
\usepackage{cite}
\usepackage[cmex10]{amsmath} 
\usepackage{color}
\usepackage[color=green!30]{todonotes}
\usepackage[letterpaper,margin=1in]{geometry}

\usepackage[ruled,vlined,linesnumbered]{algorithm2e}
\usepackage{authblk}

\usepackage{amsfonts} 
\usepackage{amssymb} 
\usepackage{mathtools}
\usepackage{amsfonts}
\usepackage{amsmath, amsthm, amssymb}
\usepackage{mathdots}
\usepackage{float}
\usepackage{booktabs}
\usepackage{caption}
\usepackage{subcaption}
\usepackage[normalem]{ulem}

\newtheorem{definition}{Definition}
\newtheorem{theorem}{Theorem}

\newtheorem{proposition}{Proposition}
\newtheorem{lemma}{Lemma}

\usepackage{subcaption}
\usepackage{tikz,pgfplots,pgf}
\usepackage{circuitikz}
\usetikzlibrary{arrows,decorations.pathreplacing}

\title{Towards Scalable Quaternary Message-Passing Decoding for Quantum Error Correction}

\author[1,2]{Boqing Zhang}
\author[2,3]{Henry D. Pfister}
\author[2]{Hanwen Yao}
\author[1]{Siyuan Niu}

\affil[1]{Department of Electrical and Computer Engineering, University of Central Florida, Orlando, FL 32816, USA}
\affil[2]{Department of Electrical and Computer Engineering, Duke University, Durham, NC 27708, USA}
\affil[3]{Department of Mathematics, Duke University, Durham, NC 27708, USA}

\date{}
\begin{document}
\maketitle
\vspace{-2em}
\begin{abstract}
\noindent
The scalability and interpretability of message-passing (MP) decoding, such as (quaternary) Belief Propagation, remain open challenges in quantum error correction. Even for surface codes, arguably the first testbed for decoding methods, studies of improved MP decoders have mostly been restricted to small distances ($d \lesssim 19$). Moreover, the mismatch with established message-passing theory limits the decoder's interpretability, making it unclear whether MP decoding can sustain its effectiveness at large system sizes.  This work takes a step toward a more principled and interpretable MP decoding framework, with the goal of making MP-based decoding more reliable and bridging theory and practice. We introduce a dilution method, which allows a quaternary Min-Sum (MS) decoder to exhibit an apparent depolarizing threshold of $16\%$ up to distance $20$, outperforming Minimum-Weight Perfect Matching in finite-length regimes. Notably, for $X$-noise, the standard MS decoder under dilution has worst-case complexity $O(N \log^2 d)$ and outperforms BP-OSD at $d=65$. The observed $\sim 9\%$ threshold may correspond to a true asymptotic threshold. 
Finally, we give a graph-dilution argument that interprets the success of the dilution method and offers insight into when MP algorithms can genuinely scale. Taken together, these results provide encouraging progress toward scalable and interpretable MP decoding in quantum error correction.
 \end{abstract}

\section{Introduction}
Message-passing (MP) algorithms such as Belief Propagation (BP) have achieved remarkable success as scalable decoders for classical low-density parity-check (LDPC) codes. When extended to quantum error correction, quaternary message passing offers additional advantages over standard decoding methods such as Minimum-Weight Perfect Matching (MWPM), as it operates directly on the joint $X$-$Z$ Tanner graph and thus has the natural potential to capture the intrinsic correlations between $X$ and $Z$ Pauli errors~\cite{poulin2008iterative}. In the fault-tolerant quantum computing era, MP decoders are also highly appealing from a hardware perspective: they are local, parallelizable, and well-suited for low-latency, on-chip implementations~\cite{muller2025improved}. 

Despite these appealing features, standard MP decoders\footnote{Here, we use the term message-passing decoder to denote decoding schemes whose computational dynamics are dominated by local iterative message updates on the decoding graph. Any auxiliary non-message-passing steps are assumed to incur only near constant ($O(1)$) overhead with respect to system size and do not affect the overall complexity scaling.}
 face severe non-convergence issues when directly applied to quantum stabilizer codes~\cite{poulin2008iterative, babar2015fifteen, raveendran2021trapping,fuentes2021degeneracy}. In the classical setting, MP algorithms like BP are exact on trees, and their asymptotic behavior on families of random classical LDPC codes, whose Tanner graphs are locally tree-like, can be rigorously characterized through information-theoretic and statistical-physics analyses~\cite{richardson2008modern, mezard2009information}. However, most quantum LDPC codes~\cite{tillich2013quantum, panteleev2021quantum, breuckmann2021balanced}, including topological codes, must embed sufficient global structure to enforce stabilizer commutativity. Hence, their Tanner graphs are highly degenerate and inevitably contain abundant short loops~\cite{poulin2008iterative}, placing them far outside this favorable tree-like regime.

For quantum codes, the behavior of BP and many variants proposed to improve it, is difficult to explain and predict from existing theory. This includes Relay-BP~\cite{muller2025improved}, BP-guided decimation~\cite{yao2024belief}, Memory BP (MBP)~\cite{kuo2022exploiting}, Generalized BP (GBP)~\cite{old2023generalized}, and many others~\cite{poulin2008iterative, wang2012enhanced, du2022stabilizer}.
Moreover, although many of the above decoders achieve good empirical performance on specific codes, there is still very little convincing numerical evidence that these MP decoders with realistic complexity constraints can sustain reliable decoding as the system size increases, with surface codes beyond distance $20$ often appearing as a canonical barrier. This stands in contrast to other non-MP decoders such as MWPM~\cite{dennis2002topological} and Union-Find~\cite{delfosse2021almost}, whose behavior is more interpretable and whose performance can scale robustly with system size~\cite{fowler2012proof, yoshida2026proof}.
For this reason, MP algorithms have been difficult to fully trust as standalone decoders, particularly in critical settings where algorithmic interpretability and conceptual clarity are essential. Furthermore, even when combined with post-processors such as Ordered Statistics Decoding (OSD)~\cite{panteleev2021degenerate} or Local Statistics Decoding (LSD)~\cite{hillmann2024localized}, evidence suggests that BP-OSD decoding performance still degrades with increasing code size~\cite{higgott2023improved}.
As the BP stage becomes increasingly fragile on larger decoding graphs, decoders with global post-processing may still face similar scalability limitations. Together, these motivate the central question of this work:

\vspace{1.5mm}
\noindent
\emph{In quantum error correction, can MP decoders achieve an asymptotic threshold with controlled iteration growth\footnote{%
By \emph{controlled iteration growth}, we mean that the maximum message-passing iteration count $I(N)$ does not grow so rapidly with $N$ as to compromise the near-linear complexity $O\big(N \cdot I(N)\big)$. Classical LDPC theory predicts $I(N)=O(\log N)$ below threshold, though such scaling may be overly optimistic for general quantum LDPC codes.} (scalability)?
If so, can their effectiveness be explained by established MP theory  (interpretability)?}
\vspace{1.5mm}

Here, we emphasize the asymptotic scaling not because practical systems require extremely large codes, but to ensure that the underlying MP mechanism is reliable and can, in principle, be sustained on large-scale decoding graphs. This question has been framed pessimistically because the Tanner graphs of quantum LDPC codes appear incompatible with the assumptions underlying established MP theory in statistical physics and spin glasses~\cite {mezard2009information, zdeborova2016statistical}. However, in this work, we show that there is room for optimism. Many quantum LDPC codes, including surface codes, can be viewed as arising from two classical tree-like codes coupled through a deterministic global operation, most notably via the topological product construction~\cite{tillich2013quantum}. 
From this viewpoint, the challenging features of their Tanner graphs can be seen as combinatorially induced by the deterministic product construction.
Motivated by this perspective, we aim to incorporate the underlying product structure into the MP decoding process. 
On the simplest surface-code model, we propose a \emph{dilution method} that first constructs a sequence of progressively diluted Tanner graphs. Each graph is structurally sparsified in a way that respects the product structure of the code and the symmetry between \(X\)- and \(Z\)-type stabilizer generators. The resulting decoder then runs a standard quaternary Min-Sum (MS) algorithm over this sequence of progressively diluted Tanner graphs. At each stage, message-passing dynamics are interleaved with collective decimation steps that freeze the qubit variables removed in the next graph. 
Numerically, this \emph{MS decoder under dilution}, with worst-case complexity $O(N \cdot \log^2 d)$, achieves a depolarizing threshold of approximately $16\%$ up to distance $20$, outperforming both MWPM and BP-OSD in this finite-length regime. As with many improved BP-based approaches, however, this threshold and finite-length advantage cannot persist for much larger systems without further modification. In contrast, for pure $X$-noise, simulations up to distance $d=65$ show no sign of performance degradation. Strong numerical evidence therefore suggests that the observed $\sim 9\%$ threshold may reflect an asymptotic threshold. This provides evidence that an MP decoding algorithm alone could sustain effective performance across the entire family of surface codes while retaining its complexity advantage.

This work also provides a framework that interprets MP decoding behavior through graph dilution. Our interpretation is built on two main conceptual components. First, we re-examine the conventional view that inevitable short cycles constitute a fundamental limitation of (quaternary) message passing, showing that this view might be incomplete without accounting for the correlation structure of the underlying error distribution and the strength of $X$-$Z$ coupling. Second, degeneracy in quantum LDPC codes renders the decoding Tanner graph locally redundant and, especially for surface codes, gives rise to a natural multiscale structure.
Combining the two perspectives, the dilution method can be interpreted as a carefully designed "renormalization step" for message passing. It coarse-grains the decoding graph, effectively zooming out to larger (sparser) scales, and induces a renormalization of the effective error through collective decimation. As a result, each diluted graph targets the renormalized error clusters at its characteristic scale, progressively correcting them all from small to large sizes. In Section~\ref{sec:diluted} and~\ref{sec:analysis}, we provide both a rigorous combinatorial analysis of the diluted graphs and a complementary heuristic interpretation of the resulting dynamics on a 1D toy model. Together, these results provide both analytical and conceptual justification for the effectiveness of the dilution method.  Overall, using the
surface code as a testbed, the main contribution of this work is to show that by exploiting the structure of the code and its decoding graph, graph dilution yields a fully deterministic message-passing decoder with \(O(N\log^2 d)\) worst-case complexity. The resulting MP decoder provides strong numerical evidence for a true asymptotic threshold and admits a conceptual interpretation.


\section{Background}
\label{sec:background}
The goal of this section is to establish a graphical model for quaternary message-passing decoding of Calderbank-Shor-Steane (CSS) codes~\cite{calderbank1996good, steane1996multiple}. While quaternary BP has been implemented and discussed in previous works \cite{kuo2020refined, lai2021log, miao2025quaternary}, we aim to develop a more systematic and principled formalization in this section.
\subsection{Quantum Decoding Problem and Graphical Models}
In quantum error correction, an $[[n,k]]$ quantum code protects $k$ logical qubits by encoding them into $n$ physical qubits that live in a subspace $\mathcal{H}_C$ of the Hilbert space $\mathcal{H}^{\otimes n}$. For stabilizer codes, the code space $\mathcal{H}_C$ is defined as the common $+1$ eigenspace of all stabilizers in the stabilizer group $\mathcal{S}$~\cite{gottesman1997stabilizer}. The stabilizer group $\mathcal{S}$ is an abelian subgroup of the $n$-fold Pauli group $\mathcal{P}_n$ with $-I^{\otimes n} \not \in \mathcal{S}$, where $\mathcal{P}_n$ is generated by tensor products of single-qubit Pauli operators $\{I, X, Y, Z\}^{\otimes n}$. The \emph{distance} $d$ of a stabilizer code $\mathcal{C}$ is defined to be the minimum weight of all Pauli operators in $N(\mathcal{S})\backslash \mathcal{S}$, where $N(\mathcal{S})$ denotes the normalizer group of $\mathcal{S}$ in $\mathcal{P}_n$. A stabilizer code $\mathcal{C}$ is called \emph{degenerate} if $d$ is larger than the minimum weight of a nontrivial stabilizer. A CSS code is a special stabilizer code whose stabilizer generators can be chosen to be either purely $X$-type or $Z$-type Pauli operators.
The surface code, introduced by Alexei Kitaev \cite{kitaev2003fault}, forms a family of simple yet highly degenerate CSS codes and is widely regarded as a promising architecture for fault-tolerant quantum computation. A distance-$d$ surface code encodes a single logical qubit using $d^2 + (d-1)^2$ physical qubits.

The quantum \emph{most-likely-error decoding} problem for a stabilizer code
$\mathcal{C}=[[n,k]]$ is the problem of inferring the underlying physical
Pauli error from the observed syndrome.
Let $\underline{E}=(E_1,\dots,E_n)\in\{I,X,Y,Z\}^n$ denote the random physical
Pauli error configuration, where for each qubit index $i\in\{1,\dots,n\}$,
the random variable $E_i$ specifies the Pauli error acting on the physical qubit
$i$. Let $\underline{e}\in\{I,X,Y,Z\}^n$ denote a generic realization of
$\underline{E}$, and let
$\underline{\sigma}=(\sigma_1,\dots,\sigma_m)$ denote the syndrome obtained
from the $m$ stabilizer measurements. Most-likely-error decoding then
corresponds to the following maximum-a-posteriori (MAP) inference over Pauli error configurations:
\vspace{-1.5mm}
\begin{equation}
    \underline{e}^*(\underline{\sigma})
    \in 
    \arg\max_{\underline{e}} \; p(\underline{e}\mid \underline{\sigma})
    =
    \arg\max_{\underline{e}} \; p(\underline{e},\underline{\sigma})
    =
    \arg\max_{\underline{e}} \; p(\underline{\sigma}\mid \underline{e})\,
    p_{\underline{E}}(\underline{e}).
    \label{eq:true_pos}
\end{equation}

\vspace{-1.5mm}
\noindent
where $p_{\underline{E}}(\underline{e})$ denotes the prior probability of the error configuration $\underline{e}$ under the physical noise model, and
$\underline{e}^*(\underline{\sigma})$ denotes a MAP-optimal estimated error, which is a maximizer of the posterior
distribution $p(\underline{e}\mid\underline{\sigma})$ conditioned on the
observed syndrome $\underline{\sigma}$. Identifying each Pauli error $e\in\{I,X,Y,Z\}$ with its binary representation $(x,z)\in\{0,1\}^2$, where $x=1$ for $e\in\{X,Y\}$ and $z=1$ for
$e\in\{Z,Y\}$, the error configuration $\underline{e}$ can equivalently be
represented as $(\underline{x},\underline{z}) \in (\mathbb{F}_2 \times \mathbb{F}_2)^n$, and the error distribution
$p_{\underline{E}}(\underline{e})$ as $p_{\underline{E}}(\underline{x},\underline{z})$. 

A \emph{Tanner graph} $\mathcal{T} = (V, F, \mathcal{E})$ of a given stabilizer code is a factor graph~\cite{poulin2008iterative, richardson2008modern}, where variable nodes $i\in V$ represent physical qubits, factor nodes $a\in F$ correspond to chosen stabilizer generators, and an edge $(i,a)\in \mathcal{E}$ is present if and only if qubit $i$ lies in the support of stabilizer $a$~\cite{mezard2009information, richardson2008modern}. For CSS codes, the factor nodes are naturally partitioned into $X$-type nodes $a^X$ and $Z$-type nodes $a^Z$, corresponding to the chosen $X$- and $Z$-stabilizer generators, respectively. For each factor node $a$, let $\partial a \subseteq V$ denote the set of neighboring variable nodes. The associated local function $\psi_a(\cdot)$ enforces the syndrome constraint of the corresponding stabilizer. Specifically, for a $Z$-type factor $a^Z$, the syndrome $\sigma_{a^Z}$ depends only on the neighboring $X$-components $\underline{x}_{\partial a^Z}$, and the corresponding local function is $\psi_{a^Z}(\underline{x}_{\partial a^Z}, \, \sigma_{a^Z}) = \mathbb{I}\!\left(\bigoplus_{i\in\partial a^Z} x_i = {\sigma}_{a^Z}\right)$. Analogously, for an $X$-type factor $a^X$, the syndrome $\sigma_{a^X}$ depends only on the neighboring $Z$-components $\underline{z}_{\partial a^X}$, and the corresponding local function is $\psi_{a^X}(\underline{z}_{\partial a^X}, \, \sigma_{a^X}) = \mathbb{I}\!\left(\bigoplus_{i\in\partial a^X} z_i = {\sigma}_{a^X}\right)$.
On each variable node $i$, we also associate a prior factor $\psi_i(e_i) = \psi_i(x_i, z_i)$, which encodes the prior Pauli error distribution on qubit $i$. Given a syndrome configuration $\underline{\sigma}$, the Tanner graph $\mathcal{T}$, together with factors $\psi_a$ and $\psi_i$, defines the following posterior distribution $\mu^{\mathcal{T}}(\underline{e} \mid \underline{\sigma})$ over error configurations $\underline{e}=\{e_i=(x_i, z_i):i\in V\}$, where each $e_i\in\{0,1\}^2$.
\vspace{-1.5mm}
\begin{equation}
       \mu^{\mathcal{T}}(\underline{e} \mid \underline{\sigma})= \mu^{\mathcal{T}}(\underline{x}, \underline{z} \mid \underline{\sigma})
      \propto 
       \biggr(\prod_{a^X \in F} \psi_{a^X}(\underline{z}_{\partial a^X}, \,   \sigma_{a^X}) \biggr)\cdot \biggr(\prod_{a^Z \in F} \psi_{a^Z}(\underline{x}_{\partial a^Z}, \,   \sigma_{a^Z}) \biggr)\cdot \biggr(\prod_{i\in V} \psi_i(x_i, z_i)\biggr).
       \label{eq:pos}
\end{equation}

\vspace{-1.5mm}
\noindent
It should be noted that, even under the code-capacity model, the distribution $\mu^{\mathcal{T}}(\underline{e} \mid \underline{\sigma})$ defined on the Tanner graph does not necessarily coincide with the true posterior $p(\underline{e} \mid \underline{\sigma})$ in Eq.~\eqref{eq:true_pos}, unless the physical error model has no spatial correlations and the single-qubit prior satisfies $p_{E_i}(x_i, z_i) \propto \psi_i(x_i, z_i)$. 
 
For surface codes, a canonical graphical representation of the Tanner graph $\mathcal{T}$ is shown in Fig.~\ref{fig:tanner}. The graph $\mathcal{T}$ naturally decomposes into two component \emph{lattices}, $\mathcal{G}_X$ and $\mathcal{G}_Z$, obtained by removing the $Z$-type and $X$-type factor nodes, respectively. In this work, we treat the Tanner graph and these component lattices as being in one-to-one correspondence: any modification of the Tanner graph is reflected in the associated lattice structure, and conversely, any structural change of the lattice induces a corresponding modified Tanner graph. This correspondence allows us to study structural dilution directly at the lattice level while retaining an equivalent representation in terms of the diluted Tanner graph.
\begin{figure}[H]
    \centering
    \scalebox{1}{
    \includegraphics[width=1\textwidth]{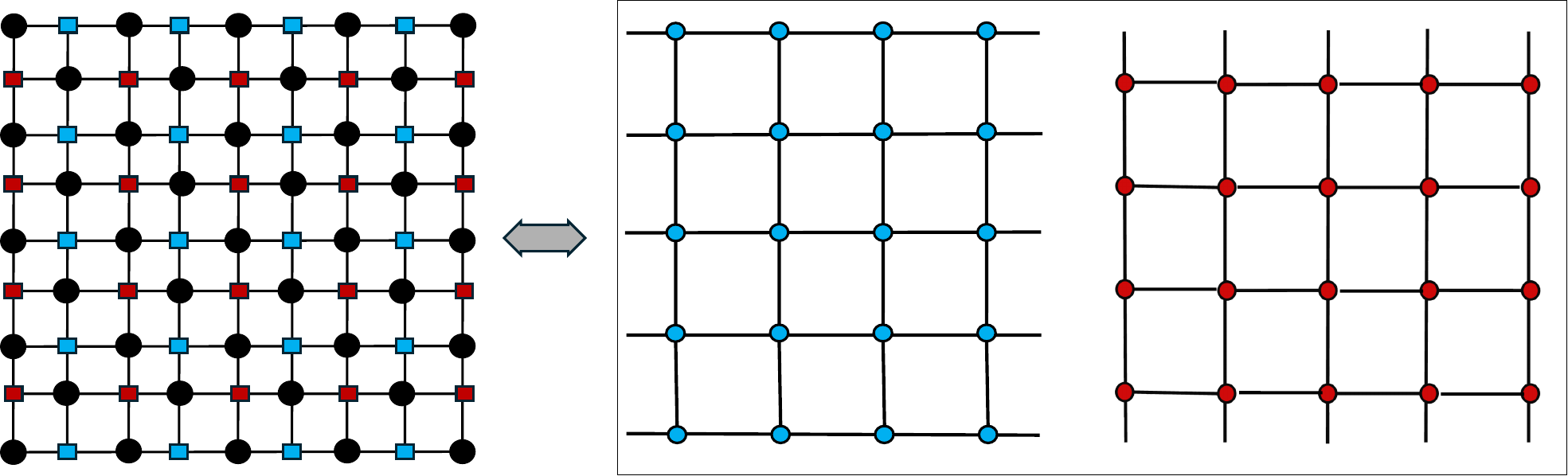} 
    }
    \vspace{-2 mm}
    \caption{Illustration of the canonical Tanner graph $\mathcal{T}$ of the distance-$5$ surface code (left) and the corresponding component lattices $\mathcal{G}_X$ (blue) and $\mathcal{G}_Z$ (red) (right). In $\mathcal{T}$, black variable nodes represent physical qubits, blue and red factor nodes correspond to $X$- and $Z$-type stabilizers, respectively. This Tanner graph decomposes into two canonical lattices in which qubits lie on edges and blue (resp. red) vertices represent $X$-type (resp. $Z$-type) stabilizers. }
    \vspace{-2 mm}
    \label{fig:tanner}
\end{figure} 
\noindent
\vspace{-8mm}
\subsection{Quaternary Message Passing for CSS Codes}
\label{subsec:mp}

After introducing the graphical model for CSS codes, we now present a quaternary message-passing decoding framework on their Tanner graphs. Given the posterior distribution $\mu^{\mathcal{T}}(\underline{e}\mid \underline{\sigma}) = \mu^{\mathcal{T}}(\underline{x}, \underline{z}\mid \underline{\sigma})$ defined on the given graph $\mathcal{T}$, the MP algorithms considered in this work iteratively exchange local messages along the edges of the graph to approximate the single-qubit posterior marginals $\mu_i^{\mathcal{T}}(e_i\mid\underline{\sigma})$ for each qubit $i$, and then construct a bitwise maximizer $e^*_{\mathcal{T}, i}(\underline{\sigma}) \in \arg\max_{e_i} \mu_i^{\mathcal{T}}(e_i \mid \underline{\sigma})$.

In the quaternary BP algorithm, given a syndrome configuration $\underline{\sigma}$, factor-to-variable message passing is employed in the Sum-Product form~\cite{richardson2008modern}. Along each edge $(i,a) \in\mathcal{E}$, the factor-to-variable message $\hat{\nu}^{(t)}_{a\to i}$ at iteration $t$ is obtained by marginalizing the local factor $\psi_a$ over its neighboring variables $\partial a\setminus i$. Owing to the CSS structure, this yields $X$-messages $\hat{\nu}^{(t)}_{a^Z\to i}(x_i)$ and $Z$-messages $\hat{\nu}^{(t)}_{a^X\to i}(z_i)$, corresponding to constraints on the $X$- and $Z$-components of $e_i=(x_i,z_i)$, respectively. As shown on the left of Fig.~\ref{fig:qbp}, the update rules are:
\begin{equation}
    \hat{\nu}_{a^Z \to i}^{(t)}(x_i)\!\propto\!\sum_{\underline{x}_{\partial a^Z \backslash i}} \psi_{a^Z}(\underline{x}_{\partial a^Z}, \sigma_{a^Z})\prod_{j \in \partial a^Z \backslash i}\!\nu^{(t)}_{j \to a^Z}(x_j), 
    \;\;
    \hat{\nu}_{a^X \to i}^{(t)}(z_i)\!\propto\!\sum_{\underline{z}_{\partial a^X \backslash i}} \psi_{a^X}(\underline{z}_{\partial a^X}, \sigma_{a^X})\prod_{j \in \partial a^X \backslash i}\!\nu^{(t)}_{j \to a^X}(z_j).
    \label{eq:f-2-v, bp}
\end{equation}
Since our goal is to identify a posterior maximizer, it is natural to replace the above marginalization step in the factor-to-variable updates by maximization. This yields the following Max–Product form~\cite{richardson2008modern}:
\begin{equation}
\hat{\nu}_{a^Z \to i}^{(t)}(x_i)\!\propto
\max_{\underline{x}_{\partial a^Z\!\setminus\! i}}
\Bigl\{\psi_{a^Z}(\underline{x}_{\partial a^Z},\sigma_{a^Z})
\!\prod_{j\in\partial a^Z\!\setminus\! i}\!\nu^{(t)}_{j\to a^Z}(x_j)\Bigr\},
\;\;
\hat{\nu}_{a^X \to i}^{(t)}(z_i)\!\propto
\max_{\underline{z}_{\partial a^X\!\setminus\! i}}
\Bigl\{\psi_{a^X}(\underline{z}_{\partial a^X},\sigma_{a^X})
\!\prod_{j\in\partial a^X\!\setminus\! i}\!\nu^{(t)}_{j\to a^X}(z_j)\Bigr\}.
\label{eq:f-2-v, ms}
\end{equation}

Note that although the above update involves both $X$ and $Z$ components, the factor-to-variable updates treat the corresponding $X$- and $Z$-messages separately. In this step, since the variables $x_i$ and $z_i$ are binary, the updates for quaternary MP retain the standard binary Sum-Product (or Max-Product) form. 

As shown on the right of Fig.~\ref{fig:qbp}, the distinction arises in the variable-to-factor message passing. For each edge $(i,a)\in \mathcal{E}$, the message $\nu_{i\to a}^{(t)}(\cdot)$ is the belief at variable node $i$ obtained by aggregating all incoming messages from neighboring factor nodes in $\partial i \setminus a$ at iteration $t$. Here, the incoming $X$- and $Z$-messages are coupled through what we refer to as an \emph{$XZ$-correlation message} $\phi_X^{(t)}(x_i)$ and $\phi_Z^{(t)}(z_i)$, which is induced by its initial prior $\psi_i(x_i, z_i)$. Thus, the quaternary variable-to-factor updates take the form:
\begin{equation}
    \nu_{i \to a^Z}^{(t+1)}(x_i) \propto \bigg(\prod_{b^Z \in \partial i \backslash a^Z} \hat{\nu}^{(t)}_{b^Z \to i}(x_i) \bigg) \cdot \phi_X^{(t)}(x_i), \quad \quad 
    \nu_{i \to a^X}^{(t+1)}(z_i) \propto \bigg( \prod_{b^X \in \partial i \backslash a^X} \hat{\nu}^{(t)}_{b^X \to i}(z_i) \bigg) \cdot \phi_Z^{(t)}(z_i).
\end{equation}
where the two specific $XZ$-correlation messages $\phi^{(t)}_X(x_i)$ and $\phi^{(t)}_Z(z_i)$ are defined as follows:
\vspace{-1mm}
\begin{equation}
\phi_X^{(t)}(x_i)
=
\sum_{z_i \in \{0,1\}}
\Bigg(
\prod_{a^X \in \partial i}
\hat{\nu}_{a^X \to i}^{(t)}(z_i)
\Bigg)\,\psi_i(x_i,z_i), 
\quad 
\phi_Z^{(t)}(z_i)
=
\sum_{x_i \in \{0,1\}}
\Bigg(
\prod_{a^Z \in \partial i }
\hat{\nu}_{a^Z \to i}^{(t)}(x_i)
\Bigg)
\, \psi_i(x_i,z_i).
\end{equation}
\begin{figure}[H]
    \centering
    \scalebox{1}{
    \includegraphics[width=1\textwidth]{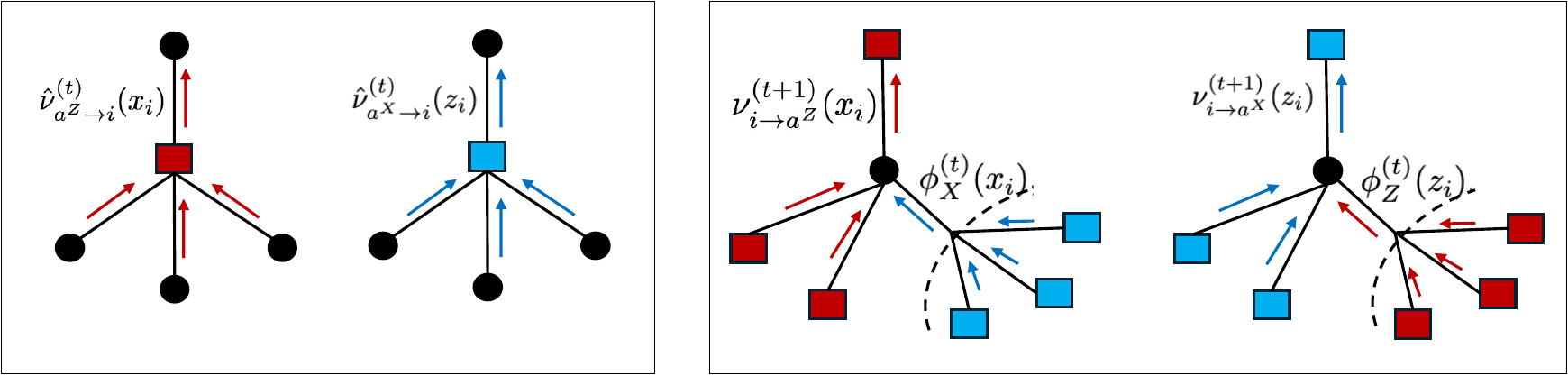} 
    }
    \vspace{-3 mm}
    \caption{Illustration of quaternary message passing for CSS codes. 
    The left panel shows the $Z$- (resp. $X$-) factor-to-variable update $\hat{\nu}_{a^Z \to i}^{(t)}(x_i)$ (resp. $\hat{\mu}_{a^X \to i}^{(t)}(z_i)$ ) from 
    $a^Z$ (resp. $a^X$) to variable node $i$, which proceeds independently and is not coupled. 
    The right panel shows the variable-to-factor ($Z$- and $X$-factor) update $\nu_{i \to a^Z}^{(t+1)}(x_i)$ and $\mu_{i \to a^X}^{(t+1)}(z_i)$, where the messages become coupled 
    through the $XZ$-correlation messages $\phi_X^{(t)}(\cdot)$ and $\phi_Z^{(t)}(\cdot)$.  }
    \vspace{-2 mm}
    \label{fig:qbp}
\end{figure} 
\noindent
After $t$ iterations, the MP \emph{estimated marginal} $\nu_i^{(t)}(e_i)=\nu_i^{(t)}(x_i, z_i)$ of qubit $i$ are obtained by combining all incoming $X$- and $Z$-messages with the initial error prior $\psi(e_i)$. Then, the MP estimated error on qubit $i$ at iteration $t$ is $\hat{e}^{(t)}_{i} \in \arg\max_{e_i} \nu^{(t)}_i(e_i)$, where $\nu_i^{(t)}(e_i)$ is:  
\begin{equation}
\nu_i^{(t)}(e_i) = \nu_i^{(t)}(x_i, z_i) =  \bigg( \prod_{a^Z \in \partial i} \hat{\nu}^{(t-1)}_{a^Z \to i}(x_i) \bigg)\cdot \bigg( \prod_{a^X \in \partial i} \hat{\nu}^{(t-1)}_{a^X \to i}(z_i) \bigg) \cdot \psi(x_i, z_i).
\label{eq:eff_marg}
\end{equation}
\noindent
In practice, message-passing decoding is implemented in the log domain for computational efficiency. Passing to the log domain transforms products of messages into sums, and the Max–Product updates into equivalent Min–Sum (MS) updates~\cite{mezard2009information}. Compared to standard BP, MS message passing is more numerically stable~\cite{panteleev2021degenerate}. In the log domain, the marginal beliefs in Eq.~\ref{eq:eff_marg} can be
conveniently parameterized by \emph{effective fields}, defined as the
log-likelihood ratios of the corresponding binary variables:
\begin{equation}
h_{i, X}^{(t)} = \log \frac{\nu_i^{(t)}(x_i=0)}{\nu_i^{(t)}(x_i=1)}, 
\qquad 
h_{i, Z}^{(t)} = \log \frac{\nu_i^{(t)}(z_i=0)}{\nu_i^{(t)}(z_i=1)}.
\end{equation}
\noindent
These effective fields serve as sufficient statistics for the binary variables and provide a convenient representation of the marginals in the log domain.

It is well known that on tree-structured graphs, Belief Propagation and Max-Product yield exact marginals and max-marginals. On graphs with loops, such as the Tanner graphs of stabilizer codes, these algorithms are generally no longer guaranteed to be exact or even converge. Nevertheless, a widely adopted empirical modification, with a long history in the literature~\cite{pretti2005message}, is to introduce a memory (or damping) factor $\epsilon$ into the message updates, as described below. This approach often improves convergence behavior without increasing the algorithmic complexity~\cite{muller2025improved, kuo2022exploiting, mooij2012sufficient, zivan2020beyond}.

\begin{equation}
\nu_{i \to a^Z}^{(t+1)}(x_i)
\;\propto\;
\prod_{b^Z \in \partial i \setminus a^Z} \hat{\nu}^{(t)}_{b^Z \to i}(x_i)\cdot \Big(\phi_X(x_i) \Big)^{1-\epsilon}
\cdot \Big(\nu_i^{(t)}(x_i)\Big)^{\epsilon}.
\end{equation}
The corresponding update for $\nu_{i \to a^X}^{(t+1)}(z_i)$ follows by symmetry, with the same damping parameter $\epsilon$.

\newpage
\section{Graph Sparsification and Dilution}
\label{sec:diluted}

This section introduces the core recipe of the dilution method. As the name suggests, the method operates by progressively diluting the decoding graph through a structured sparsification of the physical-qubit representation. In what follows, we define and focus on two important dilution patterns and highlight the distinct roles they play. We also introduce the notion of the \emph{effective error} on the diluted lattice, which will be crucial for the interpretation of the resulting message-passing dynamics in Section~\ref{sec:analysis}. 
\vspace{-1mm}
\subsection{Sparsification and Effective Error}

Informally, an $s$-sparsification removes $s$ consecutive layers of edges and retains the $(s+1)$-th layer. For surface codes, such layers can be defined in a canonical way using the coordinates of the underlying 2D lattice. Throughout this work, we refer to the parameter $s$ as the \emph{sparsification ratio}. We introduce two sparsification patterns—\emph{Diagonal} and \emph{Cartesian}—suited for depolarizing noise and pure $X$-noise, respectively.

\begin{definition}[Diagonal $s$-Sparsification]
Let $\{D_k\}_{k\in\mathbb{Z}}$ denote the ordered family of 
diagonal lines of lattice $\mathcal{G}$. 
The diagonal $s$-sparsification retains those diagonals satisfying
$
k \equiv 0 \pmod{s+1}, 
$
\noindent
and removes either all vertical or all horizontal edges from the remaining diagonals. The resulting graph is called the $D_V^s$-diluted lattice, denoted by $\mathcal{G}^s_{D_V}$, or $D^s_H$-diluted lattice, denoted by $\mathcal{G}^s_{D_H}$, depending on whether vertical or horizontal edges are removed.
\label{def:diagd}
\end{definition}
\vspace{-5mm}
\begin{figure}[H]
    \centering
    \scalebox{1}{
    \includegraphics[width=1\textwidth]{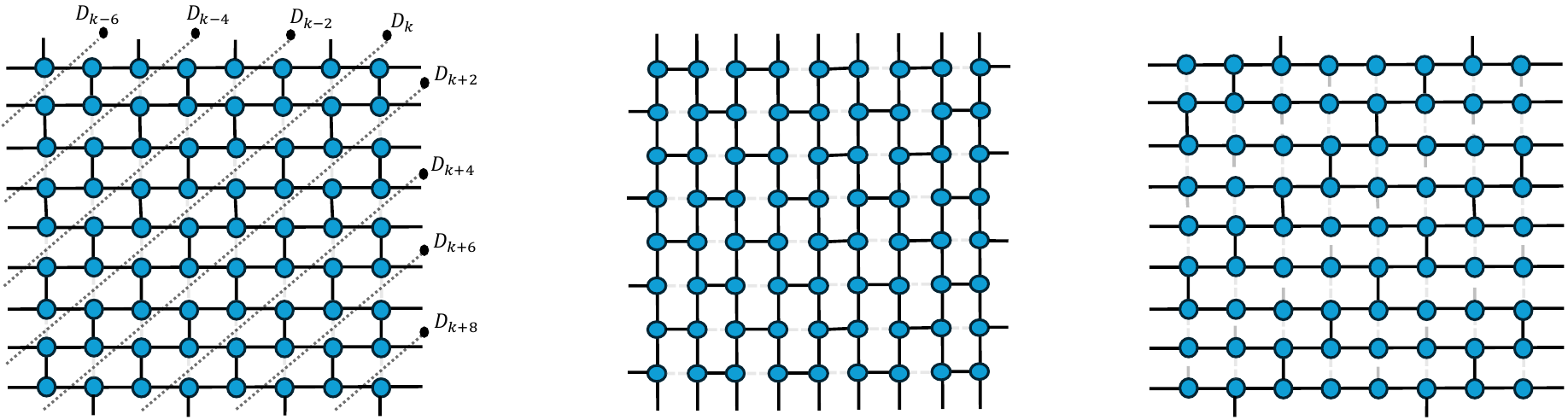} 
    }
    \vspace{-4 mm}
    \caption{Examples of diluted graphs obtained by diagonal $s$-sparsification. From left to right: the $D_V^1$-diluted lattice $\mathcal{G}^1_{D_V}$, the $D_H^1$-diluted lattice $\mathcal{G}^1_{D_H}$, and the $D_V^3$-diluted lattice $\mathcal{G}^3_{D_V}$.}
    \vspace{0 mm}
    \label{fig:diag}
\end{figure}

\begin{definition}[Cartesian $s$-Sparsification]
Let $\{C_k\}_{k\in\mathbb{Z}}$ denote the ordered family of 
Cartesian grid lines, either all horizontal or all vertical grid lines, in a lattice $\mathcal{G}$. The Cartesian $s$-sparsification retains only those grid lines satisfying
$
k \equiv 1 \pmod{s+1},
$
and removes the remaining grid lines of the same orientation. The resulting graph is called the $C_V^s$-diluted lattice, denoted by $\mathcal{G}^s_{C_V}$, or $C^s_H$-diluted lattice, denoted by $\mathcal{G}^s_{C_H}$, depending on whether horizontal or vertical grid lines are sparsified.
\label{def:horid}
\end{definition}
\vspace{-4mm}
\begin{figure}[H]
    \centering
    \scalebox{1}{
    \includegraphics[width=1\textwidth]{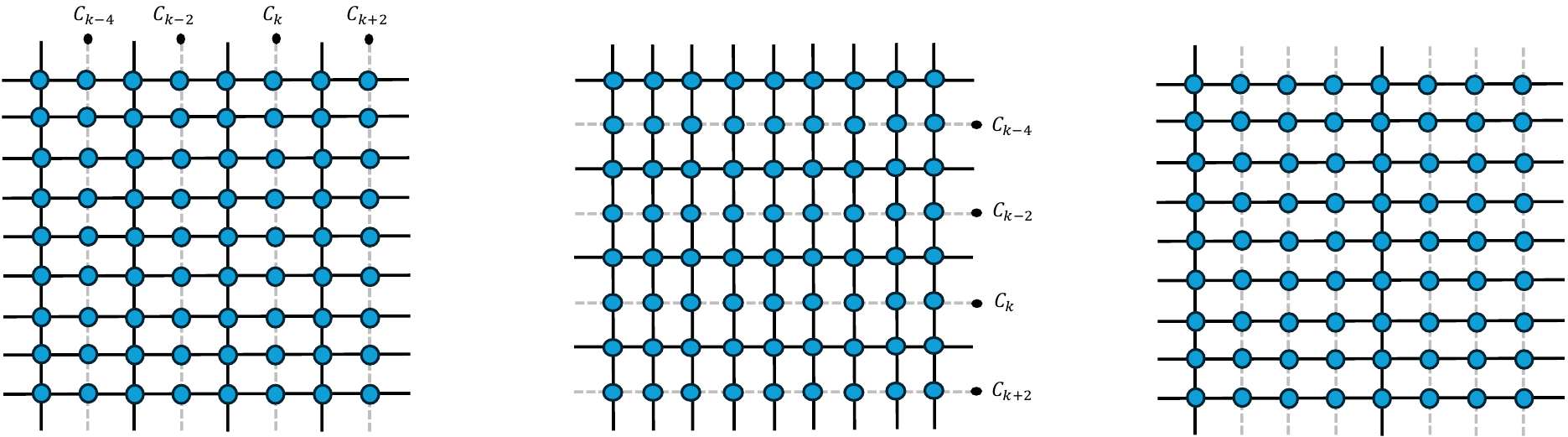} 
    }
    \vspace{-4 mm}
    \caption{Examples of diluted graphs obtained by Cartesian $s$-sparsification. From left to right: the $C_V^1$-diluted lattice $\mathcal{G}^1_{C_V}$, the $C_H^1$-diluted lattice $\mathcal{G}^1_{C_H}$,  and the $C_V^3$-diluted lattice $\mathcal{G}^3_{C_V}$.}
    \vspace{0 mm}
    \label{fig:hori}
\end{figure}
\noindent
As illustrated in Fig.\ref{fig:diag} and Fig.\ref{fig:hori}, increasing the sparsification ratio $s$ produces progressively more diluted decoding graphs. In principle, the sparsification ratio $s$ can be any nonnegative integer less than the code distance $d$. It follows from Definitions~\ref{def:diagd} and~\ref{def:horid} that the $s$-diluted lattice has girth $2s+4$. However, as $s$ increases, the error-correcting capability of the diluted graph inevitably degrades. It is therefore useful to establish a combinatorial guarantee that quantifies the error-correcting capability of an $s$-diluted graph. To quantify the error-correcting capability of a fixed diluted graph, we introduce the following definition.

\begin{definition}[Error-Correcting Radius]
Given a (diluted) decoding graph $\mathcal{G}^s \subseteq \mathcal{G}$, an \emph{ideal decoder} returns a minimum-weight correction supported on $\mathcal{G}^s$ for the observed syndrome. The \emph{error-correcting radius} of diluted graph $t(\mathcal{G}^s)$ is the largest integer $t$ such that, for every physical error $\underline e$ supported on at most $t$ qubits on $\mathcal{G}$, the correction returned by the ideal decoder operating on $\mathcal{G}^s$ is logically equivalent to $\underline e$.
\label{def:ecr}
\end{definition}

\noindent
For the surface code under single-$X$ noise, the error-correcting radius of the original $Z$-decoding lattice is half the code distance, namely $t(\mathcal{G}) = \lfloor (d-1)/2 \rfloor$. For the diluted lattices, the following theorem establishes a lower bound on the corresponding error-correcting radius for both $C^s_H$- and $D^s_H$-diluted lattices.
\begin{theorem}
\label{tm:hloss}
Fix a sparsification ratio $s \ge 1$, $s\in\mathbb{N}$, for any surface code of distance $d$, consider the diluted lattice $\mathcal{G}_{D_H}^s$ and $\mathcal{G}_{C_H}^s$ obtained from diagonal and Cartesian $s$-sparsification, respectively. Then the error-correcting radius satisfies
\begin{equation}
       t(\mathcal{G}^s_{D_H}) \geq \Big\lfloor\frac{d-1}{2\cdot \lfloor\frac{s+1}{2}\rfloor + 2} \Big \rfloor, \quad \quad 
    t(\mathcal{G}^s_{C_H}) \geq \biggr\lfloor\frac{d-1}{\lfloor\frac{s-1}{2}\rfloor + 2} \biggr \rfloor
    .
\end{equation} 
\label{tm:distance}
\end{theorem} 
\vspace{-5mm}
\begin{proof}
A detailed proof is given in Appendix~\ref{ap:1}.
\end{proof}
\noindent
From Theorem~\ref{tm:distance}, we see that for single $X$ errors, at the same sparsification ratio $s$, Cartesian $s$-sparsification gives a stronger lower bound on the error-correcting radius than diagonal $s$-sparsification, which is consistent with the numerical results in~\cite{zhang2025belief}. In particular, when $s=1$, the $C_H^1$-diluted lattice $\mathcal{G}_{C_H}^{1}$ can still provably correct all errors of weight strictly below half the distance:
$t(\mathcal{G}_{C_H}^{1})=\left\lfloor \frac{d-1}{2}\right\rfloor$ .
Thus, in the worst-case sense, the $C_H^1$-diluted lattice incurs no loss in $X$-error-correcting capability.

One justification for the diluted lattices comes from the intrinsic local stabilizer redundancy present in quantum LDPC codes. Specifically, the same syndrome pattern $\underline{\sigma}$ can admit multiple equivalent error configurations, allowing correction to be performed even on a suitably diluted decoding graph~\cite{zhang2025belief}.
A more illuminating perspective, which will be important for the interpretation in Section~\ref{sec:analysis}, 
comes from the viewpoint of the effective error configuration. The diluted lattice $\mathcal G^s$ provides a coarse-grained decoding geometry obtained by sparsifying the original lattice. On this coarse-grained lattice, and assuming we are in the correctable regime, local products of stabilizers renormalize physical error configurations $\underline e$ into effective errors $\underline e^s$ supported on the diluted lattice, as formalized in Definition~\ref{def:effective} and illustrated in Fig.~\ref{fig:effective}.

\begin{definition}[Effective Error on the Diluted Lattice]
Let $\underline{e}$ be an error configuration on the original lattice $\mathcal{G}$ with syndrome $\underline{\sigma} = \underline{\sigma}(\underline{e})$. Given a diluted lattice $\mathcal{G}^s \subseteq \mathcal{G}$, the corresponding \emph{effective error} $\underline{e}^{s} = \mathcal{R}^{s}(\underline{e})$ is defined as a minimum-weight error configuration supported on $\mathcal{G}^s$ that reproduces the same syndrome pattern, namely, $\underline{\sigma}(\underline{e}^s) = \underline{\sigma}(\underline{e})$. We call $\mathcal R^s$ the global renormalization map.
\label{def:effective}
\end{definition}
\vspace{-2mm}
\begin{figure}[H]
    \centering
    \scalebox{1}{
    \includegraphics[width=1\textwidth]{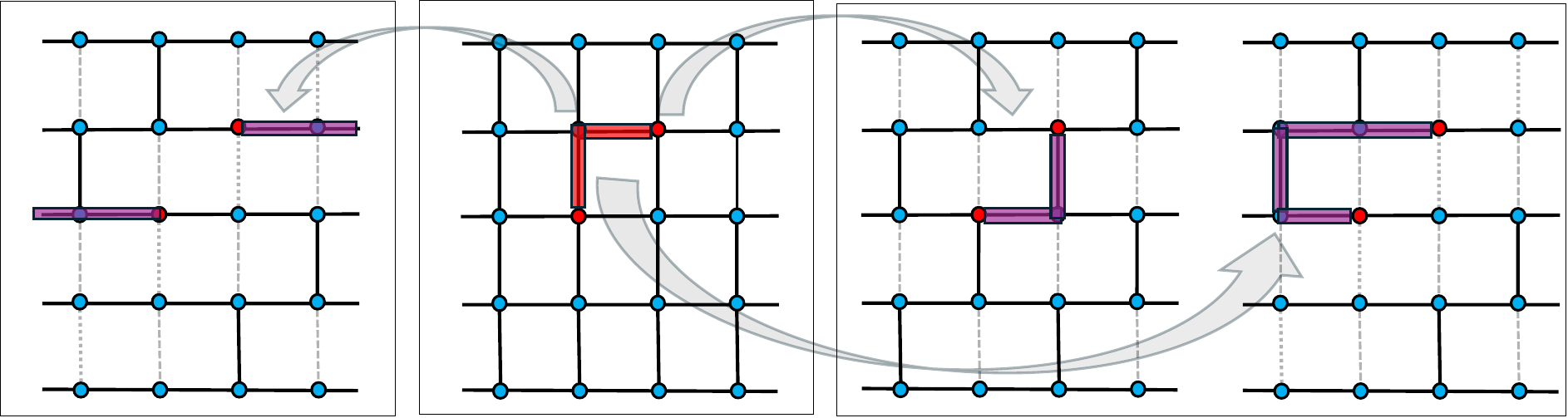} 
    }
    \vspace{-5 mm}
    \caption{Illustration of effective errors on diluted lattices. The middle panels show physical errors (red strings) on the original lattice, and their corresponding effective errors on the diluted lattices are shown in purple. The effective errors on the right are stabilizer-equivalent to the physical error, while the one on the left differs by a logical operator.}
    \label{fig:effective}
\end{figure}



\subsection{Diluted Tanner Graph and Graph-Dilution Sequence}
\label{sec:effective}

The diluted lattice is a geometric construction, whereas the quaternary message-passing decoder developed in Sec~\ref{subsec:mp} operates on a decoding Tanner graph. Thus, we first introduce a \emph{diluted Tanner graph} that translates the previously defined diluted lattice into the graphical model used for decoding.

\begin{definition}[Diluted Tanner Graph]
Consider a surface code with Tanner graph \(\mathcal{T}\). Given an \(s\)-sparsification that produces one of the diluted lattices \(\mathcal{G}^s_{D_V}\), \(\mathcal{G}^s_{D_H}\), \(\mathcal{G}^s_{C_V}\), or \(\mathcal{G}^s_{C_H}\), then the induced \emph{$s$-diluted Tanner graph} \(\mathcal{T}^s=(V^s,\mathcal{E}^s,F)\) is obtained by removing from \(\mathcal{T}\) all variable nodes corresponding to lattice edges deleted by the \(s\)-sparsification, together with their incident Tanner edges.
\label{def:diluted_tanner_graph}
\end{definition}

\begin{figure}[H]
    \centering
    \scalebox{1}{
    \includegraphics[width=1\textwidth]{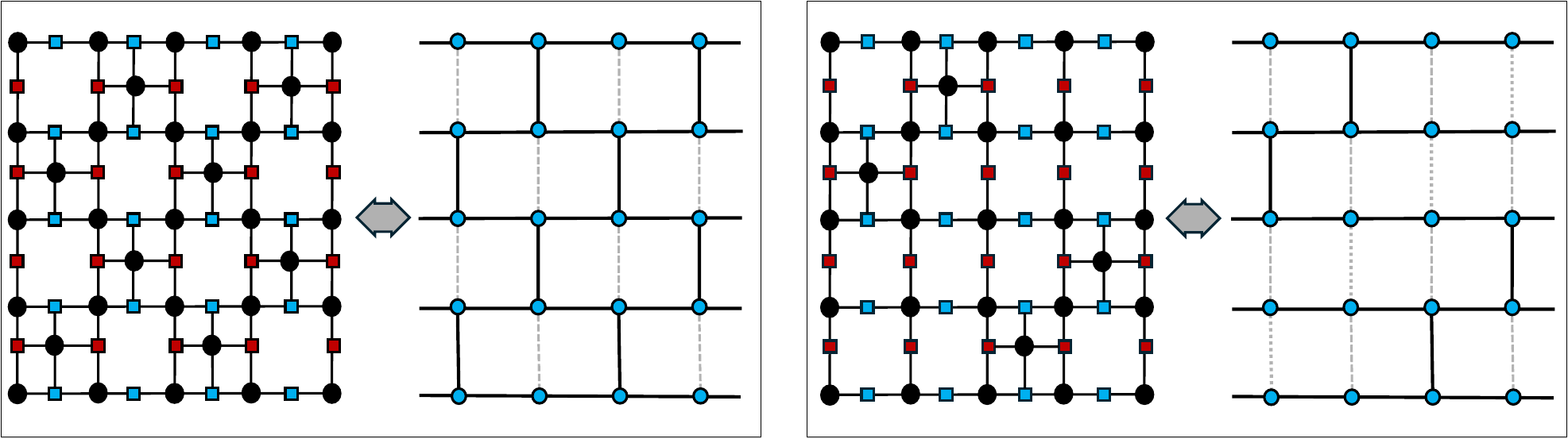} 
    }
    \vspace{-3mm}
    \caption{Examples of diluted Tanner graphs $\mathcal{T}^s$ induced by the $D_V^s$-diluted lattice $\mathcal{G}^s_{D_V}$ via diagonal sparsification pattern $D_V^s$: a $D_V^1$-diluted Tanner graph $\mathcal{T}^1_{D_V}$ (left panel), and a $D_V^3$-diluted Tanner graph $\mathcal{T}^3_{D_V}$(right panel).}
    \vspace{0 mm}
    \label{fig: diluted tanner}
\end{figure}
\noindent
 By making explicit the correspondence between the diluted Tanner graph and its component lattices, the error-correcting radius of the diluted Tanner graph admits a natural characterization in terms of the error-correcting radius of the corresponding diluted $X$- and $Z$-lattices, as shown in the following definition. 

\begin{definition}[Error-Correcting Radius of the Diluted Tanner Graph] 
\label{def: distance_loss}
Let $\mathcal{T}^s$ be a $s$-diluted Tanner graph and let $\mathcal{G}_X^s$ and $\mathcal{G}_Z^s$ be its corresponding diluted $X$- and $Z$-lattice, respectively. The corresponding error-correcting radius $t(\mathcal{T}^s)$ is defined as
\begin{equation}
t(\mathcal{T}^s)=\min \Big\{\,t\big(\mathcal{G}_X^s\big), \; t\big(\mathcal{G}_Z^s\big) \,\Big\}.
\end{equation}
where $t\big(\mathcal{G}_X^s\big)$ and $t\big(\mathcal{G}_Z^s\big)$ denote the error-correcting radius of $X$- and $Z$-lattice, respectively.
\label{def:radiu}
\end{definition}
\noindent

\noindent
Previously, Theorem~\ref{tm:distance} showed that, under single \(X/Z\) noise, the Cartesian pattern \(C_H^s\) gives a stronger lattice-level lower bound on the error-correcting radius than the diagonal pattern \(D_H^s\). However, the radius of the diluted Tanner graph, \(t(\mathcal T^s)\), depends on both component lattices; see Definition~\ref{def:radiu}. Thus, although Cartesian sparsification performs well on one component, it disconnects the other, causing \(t(\mathcal T^s)\) to vanish. In contrast, diagonal sparsification preserves the \(X\)--\(Z\) stabilizer symmetry of the surface code and therefore maintains balanced error-correcting capability against both \(X\)- and \(Z\)-type errors, as shown in Proposition~\ref{prop:symm}.

\begin{proposition}
Fix a sparsification ratio $s \in \mathbb{N}$, for any surface code of distance $d$, the error-correcting radius of its diluted Tanner $\mathcal{T}^s$ obtained from the diagonal $s$-sparsifications satisfies
\begin{equation}
       t(\mathcal{T}^s) = t\big(\mathcal{G}_X^s\big) = t\big(\mathcal{G}_Z^s\big).
\end{equation}
\label{prop:symm}
\end{proposition}
\vspace{-5mm}
\begin{proof}
Under the canonical embedding in Fig.~\ref{fig:tanner}, the $X$- and $Z$-decoding component lattices of surface code $\mathcal{G}_X$ and $\mathcal{G}_Z$ are related by a diagonal reflection. This induces a bipartite graph isomorphism
$
\phi:\; \mathcal{G}_X \xrightarrow{\cong} \mathcal{G}_Z,
$
which maps qubits to qubits, swap $X$-check with $Z$-checks, and preserving adjacency.

By Definition~\ref{def:diagd}, a diagonal $s$-sparsification removes qubits lying on diagonal
lines $D_k$ that do not satisfy the congruence condition. Since the diagonal reflection
$\phi$ naturally preserves each diagonal line, i.e., $\phi(D_k) = D_k,\, \forall k$, the set of removed (and retained) subgraphs is invariant under $\phi$. Hence, $\phi$ restricts to an isomorphism
$
\phi:\mathcal{G}_X^{\,s} \xrightarrow{\cong} \mathcal{G}_Z^{\,s}
$, and therefore
$$
t(\mathcal{G}_X^{\,s})=t(\mathcal{G}_Z^{\,s}).
$$
Finally, by Definition~\ref{def:radiu}, we have $t(\mathcal{T}^s)=\min \Big\{\,t\big(\mathcal{G}_X^s\big), \; t\big(\mathcal{G}_Z^s\big) \,\Big\} =t(\mathcal{G}_X^{\,s})=t(\mathcal{G}_Z^{\,s})$
as claimed.
\end{proof}

\noindent
From Theorem~\ref{tm:distance} and Proposition~\ref{prop:symm}, we can see that when the sparsification ratio $s$ is fixed, the error-correcting radius is lower bounded by a constant independent of distance $d$. However, as $s$ increases, this lower bound eventually vanishes. This matches the intuition that the error-correcting capability degrades as the graph is further sparsified. Therefore, running BP or MS on a single static diluted graph is not expected to work. The dilution method instead applies message passing successively to a sequence of progressively diluted decoding graphs, together with collective decimation between stages. Here, we first define \emph{graph-dilution sequence}, which produces the sequence of progressively diluted Tanner graphs induced by the corresponding diluted lattices. An illustrative example is shown in Fig.~\ref{fig: dilution_sequence}.

\begin{definition}[Graph-Dilution Sequence]
Given a distance-$d$ surface code and a fixed sparsification pattern, the corresponding graph-dilution sequence is the sequence of diluted Tanner graphs
$
\{\mathcal{T}^{(k)}\}_{k=0}^{K},
$
where $K=\lfloor\log_2 (d-1) \rfloor$. At stage-$k$, the sparsification ratio
$
s_k = 2^k - 1,
$
and the diluted graph is denoted by $\mathcal{T}^{(k)}=\mathcal{T}^{s_k}$.
\label{def:dilution-sequence}
\end{definition}
\vspace{-3mm}
\begin{figure}[H]
    \centering
    \scalebox{1}{
    \includegraphics[width=1\textwidth]{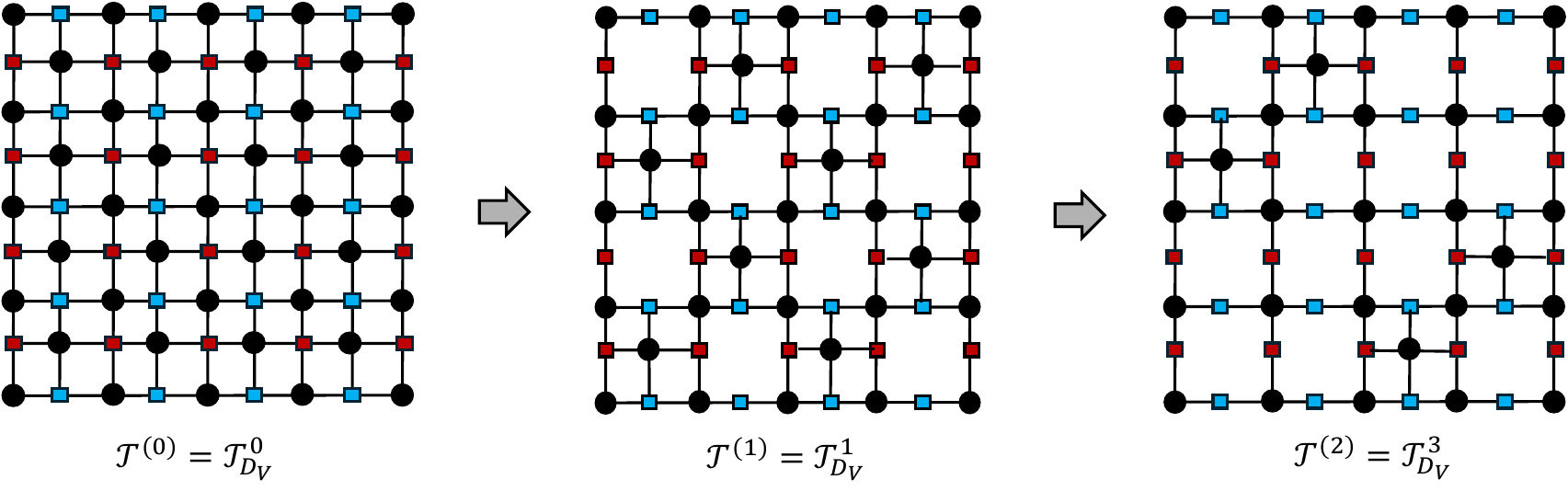} 
    }
    \vspace{-5mm}
    \caption{An example of the graph-dilution sequence \(\{\mathcal{T}^{(0)}, \mathcal{T}^{(1)}, \mathcal{T}^{(2)}\}\) for a distance-\(5\) surface code. The sparsification pattern is the diagonal pattern \(D_V^s\).}
    \vspace{0 mm}
    \label{fig: dilution_sequence}
\end{figure}
\noindent
In this section, given a distance-$d$ surface code and a sparsification pattern ($D_V$, $D_H$, $C_V$ or $C_H$), we show how the dilution method constructs the corresponding graph-dilution sequence $\{\mathcal{T}^{(k)}\}_{k=0}^{K}$. Furthermore, Proposition~\ref{prop:symm} and Theorem~\ref{tm:distance} also suggest a natural guideline for choosing the sparsification pattern under different noise models: diagonal patterns, corresponding to the $D_V$- or $D_H$-graph-dilution sequence, are better suited for quaternary MP decoding under balanced or correlated $X$-$Z$ noise such as depolarizing noise; Cartesian patterns, corresponding to the $C_V$- or $C_H$-graph-dilution sequence, are more appropriate for single-type or highly biased $X/Z$ noise.

\section{Algorithms and Numerical Results}
\label{sec:algnum}

After introducing the construction of the graph-dilution sequence, we now describe the \emph{dilution method} for a general class of quaternary message-passing decoders. The algorithmic core is simple and concrete: rather than running message passing directly on the original decoding graph, we run it on the progressively diluted Tanner graphs defined in the graph-dilution sequence, using
\(I_k = 20k+20\) iterations at stage \(k\). Between successive dilution stages, we perform a global hard-decimation step: any qubit removed by the next dilution stage is frozen to $\{I, X, Y, Z\}$ according to its current effective field $h^{(t)}_{i,X}$ and $h^{(t)}_{i,Z}$ estimated by the quaternary message passing. The specific graph-dilution sequence depends on the noise model: $C^s_H$- and $C^s_V$-dilution sequence for single $X$-noise, whereas $D^s_H$- and $D^s_V$-dilution sequence for depolarizing noise. In both cases, the algorithmic structure remains unchanged.

\begin{algorithm}[htbp] 
\caption{Quaternary Message-Passing under Dilution}
\label{alg:bp_decoding}
\KwIn{Syndrome \(\underline{\sigma}\); surface-code distance $d$; chosen sparsification pattern (\(D_V\), \(D_H\), \(C_V\), or \(C_H\))}
Construct the corresponding graph dilution sequence \(\{\mathcal{T}^{(k)}\}_{k=0}^{K(d)}\)\;
\For{$i = 0,1,\dots,K-1$}{
Reset and run quaternary MP on diluted Tanner graph \(\mathcal{T}^{(k)}\) for $I_k=20k+20$ iterations\;
Decimate and freeze the qubits supported on \(\mathcal{T}^{s_i} \backslash \mathcal{T}^{s_{i+1}}\)\;
}
\KwOut{Final hard decision in $\{I,X,Y,Z\}^n$ upon convergence}
\end{algorithm}
\noindent

\subsection{Simulation Results on the Surface-Code family}

Since Min-Sum (MS) is more numerically stable than Sum-Product BP, we use quaternary MS decoding with fixed memory factor $\epsilon$ in practical simulation. We evaluate the method under two representative noise models. For the single-$X$ noise model, each qubit undergoes an independent $X$ error with probability $p$, yielding the uniform prior
$
p_{E_i}(e_i) = [1-p,\; p,\; 0,\; 0].
$
For the depolarizing noise model, each qubit experiences an $X$, $Y$, or $Z$
error with equal probability $p/3$, giving the prior
$
p_{E_i}(e_i) = [1-p,\; p/3,\; p/3,\; p/3].
$
In both cases, the error prior is set according to the physical model,
$\psi_i(e_i) = p_{E_i}(e_i)$.

We now turn to the numerical results for decoding performance under depolarizing noise, shown in Fig.~\ref{fig:bp4}. In the finite-length regime $d<20$, the quaternary Min-Sum decoder under dilution outperforms MWPM, as shown in the right panel of Fig.~\ref{fig:bp4}. The left panel shows an apparent threshold of approximately 16\% over the range of distances. However, both the finite-length advantage and the apparent threshold weaken as the system size increases. In particular, at distance $d=33$, the logical error rate no longer decreases and instead begins to increase, as shown in the left panel of Fig.~\ref{fig:bp4}. 
\begin{figure}[H]
\centering
\begin{tikzpicture}
\definecolor{darkgray176}{RGB}{176,176,176}
\definecolor{lightgray204}{RGB}{204,204,204}
\definecolor{color1}{rgb}{0.0000,0.4470,0.7410}
\definecolor{color2}{rgb}{0.8500,0.3250,0.0980}
\definecolor{color3}{rgb}{0.9290,0.6940,0.1250}
\definecolor{color4}{rgb}{0.4940,0.1840,0.5560}
\definecolor{color5}{rgb}{0.4660,0.6740,0.1880}
\definecolor{color6}{rgb}{0.3010,0.7450,0.9330}
\definecolor{color7}{rgb}{0.6350,0.0780,0.1840}
\definecolor{color8}{rgb}{0.8350,0.0880,0.1840}

\begin{axis}[
width=0.48\linewidth,
height=0.42\linewidth,
legend cell align={left},
legend columns = 1,
legend style={
  fill opacity=0.7,
  draw opacity=1,
  text opacity=1,
  at={(0.85,-0.30)},
  anchor=south east,
  draw=lightgray204,
  font=\scriptsize
},
grid=both,
ticklabel style = {font=\footnotesize},
label style = {font=\footnotesize},
x grid style={darkgray176},
xtick={0.01,0.02,0.03,0.04,0.05,0.06,0.07,0.08,0.09,0.10,0.11,0.12,0.13,0.14,0.15,0.16,0.17},
xticklabels={,0.02,,0.04,,0.06,,0.08,,0.10,,0.12,,0.14,,0.16,},
xmin=0.01,xmax=0.17,
ymin=1e-5,ymax=0.4,
ymode=log,
y grid style={darkgray176},
ylabel={Total Error Rate},
  ylabel style={
    at={(axis description cs: 0.05, 0.25)},
    anchor=south,
  },
yscale=0.77
],
yscale=0.77
]

\addplot[thick, color3, mark=*, 
mark options={scale=1}, smooth] 
table[x=p, y expr = \thisrow{errors}/10000] {
p errors
0.01  6
0.02  34
0.03  91
0.04  179
0.05  301
0.06  407
0.07  554
0.08  744
0.09  925
0.1   1109
0.11  1344
0.12  1527
0.13  1772
0.14  1980
0.15  2307
0.16  2576
0.17  2700
};
\addlegendentry{d=3}

9
54
91
179
301
407
554
744
925
1109
1344
1527
1772
1980
2307
2576
2700


\addplot[thick, color2, mark=*, 
mark options={draw=color2, scale=1}, smooth] 
table[x=p, y expr = \thisrow{errors}/10000] {
p errors
0.01  0.1
0.02  2.8
0.03  12.7
0.04  30.6
0.05  68.6
0.06  113.2
0.07  188.9
0.08  291
0.09  428
0.1   579
0.11  755
0.12  976
0.13  1243
0.14  1557
0.15  2179
0.16  2438
0.17  2715
};
\addlegendentry{d=5}

0.1
0.7
4
11
34
62
101
197
311
530
731
954
1243
1557
2279
2538
2715

\addplot[thick, color1, mark=*, 
mark options={draw=color1, scale=1}, smooth] 
table[x=p, y expr = \thisrow{errors}/10000] {
p errors
0.01  0
0.02  0
0.03  0.4
0.04  2
0.05  6
0.06  14
0.07  42
0.08  72
0.09  120
0.1   243
0.11  422
0.12  639
0.13  897
0.14  1226
0.15  1717
0.16  2497
0.17  3097
};
\addlegendentry{d=9}


0
0
0.2
1.5
7.1
18.7
38
79
162
353
530
794
1019
1495
2156
2597
3097


\addplot[thick, color5, mark=*, 
mark options={draw=color5, scale=1}, smooth] 
table[x=p, y expr = \thisrow{errors}/10000] {
p errors
0.01  0
0.02  0
0.03  0
0.04  0
0.05  0
0.06  1
0.07  6
0.08  25
0.09  52
0.1   102
0.11  234
0.12  408
0.13  741
0.14  1127
0.15  1840
0.16  2510
0.17  3228
};
\addlegendentry{d=17}

0
0
0
0
0.5
2.2
8
23
69
118
256
451
803
1187
1729
2455
3128

\addplot[thick, color7, mark=*, 
mark options={draw=color7, scale=1}, smooth] 
table[x=p, y expr = \thisrow{errors}/10000] {
p errors
0.01  0
0.02  0
0.03  0
0.04  0
0.05  0
0.06  0
0.07  1.4
0.08  8
0.09  31
0.1   81
0.11  192
0.12  362
0.13  706
0.14  1135
0.15  1703
0.16  2443
0.17  3491
};
\addlegendentry{d=19}

0
0
0
0
0
0
2
15
30
82
201
373
717
1156
1715
2484
3491



\addplot[black, line width=1.5pt, dash pattern=on 3pt off 3pt, mark=o, mark options={draw=black, scale=1}, smooth] 
table[x=p, y expr = \thisrow{errors}/10000] {
p errors

0.09  500
0.1   1200
0.11  1500
0.12  1800
0.13  2200
0.14  2700
0.15  3100
0.16  5200
0.17  6800
};
\addlegendentry{d=33}

\addplot+[
  dashed,
  red,
  mark options={scale=0.25},
  line width=1pt 
] coordinates {
  (0.16,1e-5)
  (0.16,1)
};

\end{axis}
\end{tikzpicture}
\hspace{2mm}
\begin{tikzpicture}
\definecolor{darkgray176}{RGB}{176,176,176}
\definecolor{lightgray204}{RGB}{204,204,204}
\definecolor{color1}{rgb}{0.0000,0.4470,0.7410}
\definecolor{color2}{rgb}{0.8500,0.3250,0.0980}
\definecolor{color3}{rgb}{0.9290,0.6940,0.1250}
\definecolor{color4}{rgb}{0.4940,0.1840,0.5560}
\definecolor{color5}{rgb}{0.4660,0.6740,0.1880}
\definecolor{color6}{rgb}{0.3010,0.7450,0.9330}
\definecolor{color7}{rgb}{0.6350,0.0780,0.1840}
\definecolor{color8}{rgb}{0.8350,0.0880,0.1840}

\begin{axis}[
width=0.48\linewidth,
height=0.42\linewidth,
legend cell align={left},
legend columns = 1,
legend style={
  fill opacity=0.7,
  draw opacity=1,
  text opacity=1,
  at={(1,-0.35)},
  anchor=south east,
  draw=lightgray204,
  font=\scriptsize
},
grid=both,
ticklabel style = {font=\footnotesize},
label style = {font=\footnotesize},
x grid style={darkgray176},
xtick={0.01,0.02,0.03,0.04,0.05,0.06,0.07,0.08,0.09,0.10,0.11,0.12,0.13,0.14,0.15, 0.16,0.17},
xticklabels={,0.02,,0.04,,0.06,,0.08,,0.10,,0.12,,0.14,,0.16,},
xmin=0.03,xmax=0.15,
ymin=1e-5,ymax=0.2,
ymode=log,
y grid style={darkgray176},
yscale=0.77
]



\addplot[thick, dashed, color1, mark=o, 
mark options={scale=1, solid}, smooth] 
table[x=p, y expr = \thisrow{errors}/10000] {
p errors
0.01  0
0.02  0
0.03  1
0.04  4
0.05  20
0.06  52
0.07  102
0.08  185
0.09  377
0.1   579
0.11  897
0.12  1232
0.13  1660
0.14  2112
0.15  2550
0.16  3096
0.17  3568
};
\addlegendentry{d=9, MWPM}

\addplot[thick, color1, mark=*, 
mark options={draw=color1, scale=1}, smooth] 
table[x=p, y expr = \thisrow{errors}/10000] {
p errors
0.01  0
0.02  0
0.03  0.4
0.04  2
0.05  6
0.06  14
0.07  42
0.08  72
0.09  120
0.1   243
0.11  422
0.12  639
0.13  897
0.14  1226
0.15  1717
0.16  2497
0.17  3097
};
\addlegendentry{d=9, MS-Dilution}

\addplot[thick, dashed, color7, mark=o, 
mark options={scale=1, solid}, smooth] 
table[x=p, y expr = \thisrow{errors}/10000] {
p errors
0.01  0
0.02  0
0.03  0
0.04  0.1
0.05  1
0.06  4.2
0.07  15
0.08  41
0.09  98
0.1   225
0.11  439
0.12  765
0.13  1192
0.14  1748
0.15  2404
0.16  3169
0.17  3896
};
\addlegendentry{d=15, MWPM}

\addplot[thick, color7, mark=*, 
mark options={draw=color7, scale=1}, smooth] 
table[x=p, y expr = \thisrow{errors}/10000] {
p errors
0.01  0
0.02  0
0.03  0
0.04  0
0.05  0.4
0.06  2.2
0.07  8
0.08  25
0.09  61
0.1   132
0.11  256
0.12  451
0.13  803
0.14  1187
0.15  1729
0.16  2455
0.17  3128
};
\addlegendentry{d=15, MS-Dilution}


0
0
0
0.3
1.8
9
24
67
149
284
523
912
1364
1998
3006
3454
4212

0
0.1
2
21
111
488
1555
3876
7680
13197

\end{axis}
\end{tikzpicture}
\hspace{2mm}
\caption{Performance of quaternary Min-Sum decoding with dilution under depolarizing noise (left), together with a finite-length comparison against the standard global decoder MWPM at distances $9$ and $15$ (right). For MS under dilution, the damping factor is \(\epsilon=0.15\) for all distances.}
\label{fig:bp4}
\end{figure}
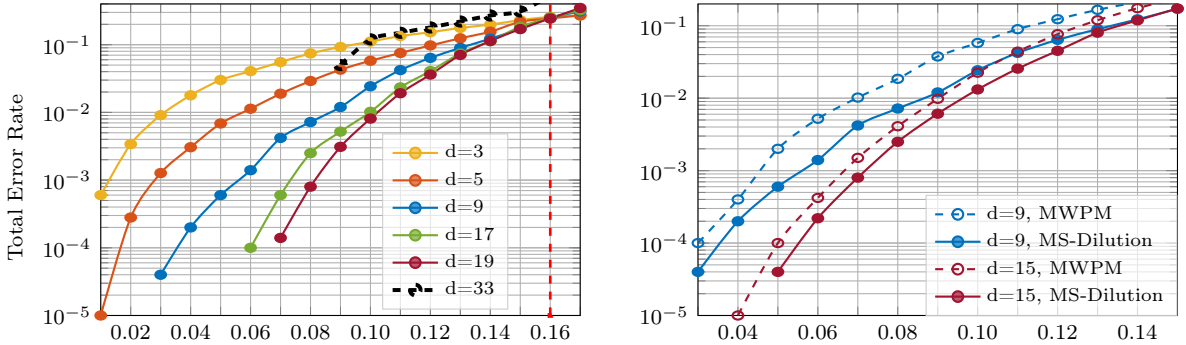
\vspace{-1mm}
\noindent
Under $X$-noise, the situation is different, as shown in Fig.~\ref{fig:bp2}. Message-passing algorithms like BP and MS have previously been applied to the surface codes, but existing studies have provided limited evidence for a true asymptotic threshold.
Figure~\ref{fig:bp2} shows that Min-Sum decoding under dilution continues to
perform well up to \(d=65\). The curves exhibit a clear threshold near \(9\%\)
with no visible finite-size degradation. To our knowledge, this provides the
first numerical evidence that a simple Min-Sum decoder can sustain threshold
behavior on surface codes at large system sizes, e.g., \(d=65\) with
\(N\approx 8300\), while retaining the complexity advantages of local
message-passing algorithms; see Sec.~\ref{sec:icc}. These results give strong
evidence that dilution can make message passing reliable across the
surface-code family.
The finite-length advantage is also visible in the right panel of Fig.~\ref{fig:bp2}. Since the dilution method makes message passing scalable, Min-Sum under dilution outperforms the global decoder MS-OSD at large system sizes,
for example, at \(d=65\).
\vspace{-1mm}
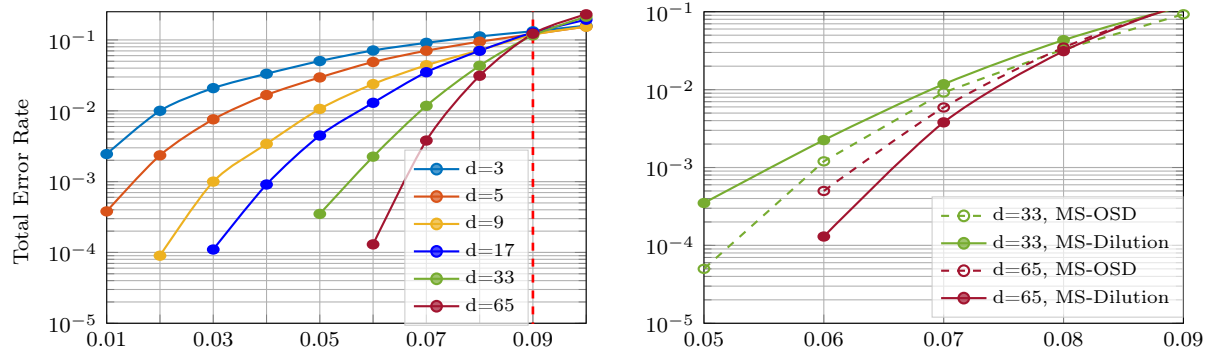
\begin{figure}[H]
\centering
\begin{tikzpicture}
\definecolor{darkgray176}{RGB}{176,176,176}
\definecolor{lightgray204}{RGB}{204,204,204}
\definecolor{color1}{rgb}{0.0000,0.4470,0.7410}
\definecolor{color2}{rgb}{0.8500,0.3250,0.0980}
\definecolor{color3}{rgb}{0.9290,0.6940,0.1250}
\definecolor{color4}{rgb}{0.4940,0.1840,0.5560}
\definecolor{color5}{rgb}{0.4660,0.6740,0.1880}
\definecolor{color6}{rgb}{0.3010,0.7450,0.9330}
\definecolor{color7}{rgb}{0.6350,0.0780,0.1840}
\definecolor{color8}{rgb}{0.8350,0.0880,0.1840}

\begin{axis}[
width=0.48\linewidth,
height=0.42\linewidth,
legend cell align={left},
legend columns = 1,
legend style={
  fill opacity=0.7,
  draw opacity=1,
  text opacity=1,
  at={(0.88,-0.35)},
  anchor=south east,
  draw=lightgray204,
  font=\scriptsize
},
grid=both,
ticklabel style = {font=\footnotesize},
label style = {font=\footnotesize},
x grid style={darkgray176},
xtick={0.01, 0.02, 0.03, 0.04, 0.05, 0.06, 0.07, 0.08, 0.09, 0.10},
xticklabels={0.01,,0.03,,0.05,,0.07,,0.09},
xmin=0.01,xmax=0.10,
ymin=1e-5,ymax=0.25,
ymode=log,
y grid style={darkgray176},
ylabel={Total Error Rate},
  ylabel style={
    at={(axis description cs: 0.05, 0.25)},
    anchor=south,
  },
yscale=0.77
]

\addplot[thick, color1, mark=*, 
mark options={scale=1}, smooth] 
table[x=p, y expr = \thisrow{errors}/100000] {
p errors
0.01  245
0.02  999
0.03  2087
0.04  3330
0.05  5034
0.06  7099
0.07  9096
0.08  11222
0.09  13252
0.1   15901
};
\addlegendentry{d=3}

245
999
2087
3330
5034
7099
9096
11222
13252
15901

\addplot[thick, color2, mark=*, 
mark options={scale=1}, smooth] 
table[x=p, y expr = \thisrow{errors}/100000] {
p errors
0.01  38
0.02  235
0.03  757
0.04  1669
0.05  2955
0.06  4886
0.07  7043
0.08  9449
0.09  12080
0.1   15412
};
\addlegendentry{d=5}

\addplot[thick, color3, mark=*, 
mark options={scale=1}, smooth] 
table[x=p, y expr = \thisrow{errors}/100000] {
p errors
0.01  0
0.02  9
0.03  100
0.04  341
0.05  1062
0.06  2382
0.07  4390
0.08  7159
0.09  11836
0.1   15394
};
\addlegendentry{d=9}


\addplot[thick, blue, mark=*, 
mark options={ scale=1}, smooth] 
table[x=p, y expr = \thisrow{errors}/100000] {
p errors
0.01  0
0.02  0
0.03  11
0.04  91
0.05  447
0.06  1296
0.07  3500
0.08  7027
0.09  12716
0.1   19283
};
\addlegendentry{d=17}



0
0
0
8
88
493
1824
5278
11343
23390


\addplot[thick, color5, mark=*, 
mark options={ scale=1}, smooth] 
table[x=p, y expr = \thisrow{errors}/100000] {
p errors
0.01  0
0.02  0
0.03  0
0.04  0
0.05  35
0.06  225
0.07  1174
0.08  4320
0.09  11588
0.1   21260
};
\addlegendentry{d=33}

\addplot[thick, color7, mark=*, 
mark options={scale=1}, smooth] 
table[x=p, y expr = \thisrow{errors}/100000] {
p errors
0.01  0
0.02  0
0.03  0
0.04  0
0.05  0
0.06  13
0.07  380
0.08  3120
0.09  12250
0.1   23000
};
\addlegendentry{d=65}

\addplot+[
  dashed,
  red,
  mark options={scale=0.25},
  line width=1pt 
] coordinates {
  (0.090,1e-6)
  (0.090,1)
};

\end{axis}
\end{tikzpicture}
\hspace{2mm}
\begin{tikzpicture}
\definecolor{darkgray176}{RGB}{176,176,176}
\definecolor{lightgray204}{RGB}{204,204,204}
\definecolor{color1}{rgb}{0.0000,0.4470,0.7410}
\definecolor{color2}{rgb}{0.8500,0.3250,0.0980}
\definecolor{color3}{rgb}{0.9290,0.6940,0.1250}
\definecolor{color4}{rgb}{0.4940,0.1840,0.5560}
\definecolor{color5}{rgb}{0.4660,0.6740,0.1880}
\definecolor{color6}{rgb}{0.3010,0.7450,0.9330}
\definecolor{color7}{rgb}{0.6350,0.0780,0.1840}
\definecolor{color8}{rgb}{0.8350,0.0880,0.1840}

\begin{axis}[
width=0.48\linewidth,
height=0.42\linewidth,
scaled x ticks=false,
legend cell align={left},
legend columns = 1,
legend style={
  fill opacity=0.7,
  draw opacity=1,
  text opacity=1,
  at={(1,-0.35)},
  anchor=south east,
  draw=lightgray204,
  font=\scriptsize
},
grid=both,
ticklabel style = {font=\footnotesize},
label style = {font=\footnotesize},
x grid style={darkgray176},
xtick={0.02, 0.03, 0.04, 0.05, 0.06, 0.07, 0.08, 0.09},
xticklabels={0.02,0.03,0.04,0.05,0.06,0.07,0.08,0.09},
xmin=0.05,xmax=0.09,
ymin=1e-5,ymax=1e-1,
ymode=log,
y grid style={darkgray176},
yscale=0.77
]

\addplot[thick, color5, dashed, mark=o, 
mark options={scale=1, solid}] 
table[x=p, y expr = \thisrow{errors}/100000] {
p errors
0.03  0
0.04  0
0.05  5
0.06  120
0.07  920
0.08  3290
0.09  9280
};
\addlegendentry{d=33, MS-OSD}

\addplot[thick, color5, mark=*, 
mark options={ scale=1}, smooth] 
table[x=p, y expr = \thisrow{errors}/100000] {
p errors
0.01  0
0.02  0
0.03  0
0.04  0
0.05  35
0.06  225
0.07  1174
0.08  4320
0.09  11588
0.1   21260
};
\addlegendentry{d=33, MS-Dilution}

\addplot[thick, color7, dashed, mark=o, 
mark options={scale=1.0, solid}, smooth] 
table[x=p, y expr = \thisrow{errors}/100000] {
p errors
0.03  0
0.04  0
0.05  0
0.06  50
0.07  590
0.08  3470
0.09  12700
};
\addlegendentry{d=65, MS-OSD}

\addplot[thick, color7, mark=*, 
mark options={scale=1}, smooth] 
table[x=p, y expr = \thisrow{errors}/100000] {
p errors
0.01  0
0.02  0
0.03  0
0.04  0
0.05  0
0.06  13
0.07  380
0.08  3120
0.09  12250
0.1   23000
};
\addlegendentry{d=65, MS-Dilution}

0
0
0
3
37
237
1190
3460
8630
14617


0
0.3
4
39
180
661
2028
4653
8943
14617

\end{axis}
\end{tikzpicture}
\hspace{2mm}
\caption{Performance of Min-Sum decoding with dilution under $X$-noise (left), together with a finite-length comparison against the standard global decoders MS-OSD under $X$-noise at distances $33$ and $65$ (right). For MS-OSD, the number of Min-Sum iterations is set to \(30\). For MS under dilution, the damping factor is \(\epsilon=0.1\) for \(d=33\), \(\epsilon=0.05\) for \(d=65\), and \(\epsilon=0.15\) for all other distances.}
\label{fig:bp2}
\end{figure}
\vspace{-1mm}

\vspace{-1mm}
\subsection{Complexity and Convergence Discussion}
\label{sec:icc}
\noindent
In the classical literature, BP on sparse random graphs typically converges in $O(\log N)$ iterations~\cite{richardson2008modern}. In this work, we choose \(I_k=20k+20\). Hence the total number of message-passing iterations over all stages is
$$
    I_{\max}
    =
    \sum_{k=0}^{K} I_k
    =
    \sum_{k=0}^{\lfloor \log_2 d \rfloor} (20k+20)
    =
    10(K+1)(K+2)=O(\log^2 d),
$$
where \(K=\lfloor \log_2 d \rfloor\). The resulting worst-case complexity is
$
    O(N \cdot I_{\max}) = O(N \cdot \log^2 d).
$
In the average-case scenario, the number of iterations is often much smaller, especially in the low-error-rate regime, as shown in the left panel of Fig.~\ref{fig:complexity}. Thus, the average-case complexity remains close to
linear in \(N\).

\begin{figure}[H]
\centering
\begin{tikzpicture}
\definecolor{darkgray176}{RGB}{176,176,176}
\definecolor{color1}{rgb}{0.0000,0.4470,0.7410}

\begin{axis}[
  name=main,
  width=0.48\linewidth,
  height=0.35\linewidth,
  xmin=9,xmax=65,
  ymin=10,ymax=150,
  xlabel={Code Distance $d$},
  ylabel={Number of Iterations $I$},
  ylabel style={
    at={(axis description cs: 0.08, 0.5)},
    anchor=south,
  },
  grid=both,
  x grid style={darkgray176},
  y grid style={darkgray176},
  ticklabel style={font=\footnotesize},
  label style={font=\footnotesize},
  xtick={10, 20, 30, 40, 50, 60, 70},
  xticklabels={10, 20, 30,  40, 50, 60, 70},
  ytick={10, 50, 100, 150},
  legend cell align={left},
  legend style={
    fill opacity=0.7,
    draw opacity=1,
    text opacity=1,
    at={(0.45,0.35)},
    anchor=north west,
    font=\scriptsize
  },
]

\addplot[thick, dashed, red, mark=o,
  mark options={draw=red, scale=1.1, solid}, smooth]
table[x=d, y=num_of_iter] {
d  num_of_iter
9  12.0
17 13.5
33 15.8
65 17.5
};
\addlegendentry{Typical $I$ (p=0.03)}

\addplot[thick, gray, mark=*,
  mark options={draw=gray, scale=1.0}, smooth]
table[x=d, y=num_of_iter] {
d  num_of_iter
9  18.0
17 49
33 72.0
65 98.0
};
\addlegendentry{Typical $I$ (p=0.06)}

\end{axis}

\begin{axis}[
  width=0.24\linewidth,
  height=0.20\linewidth,
  at={(main.south east)},          
  xshift=-3.2cm, yshift=2.28cm,   
  anchor=south east,
  xmin=3,xmax=65,
  ymin=20,ymax=500,
  grid=both,
  x grid style={darkgray176},
  y grid style={darkgray176},
  ticklabel style={font=\scriptsize},
  label style={font=\scriptsize},
  xtick={3, 23, 43, 63},
  ytick={20,220,420},
  legend cell align={left},
  legend style={
    fill opacity=0.7,
    draw opacity=1,
    text opacity=1,
    at={(0.00,1)},
    anchor=north west,
    font=\scriptsize
  },
]

\addplot[thick, black, mark=*,
  mark options={draw=black, scale=0.8}, smooth]
table[x=d, y=num_of_iter] {
d  num_of_iter
3  20
5  60
9  120
17 200
33 300
65 420
};
\addlegendentry{Max $I$}

\end{axis}

\end{tikzpicture}
\hspace{2mm}
\begin{tikzpicture}
\definecolor{darkgray176}{RGB}{176,176,176}
\definecolor{lightgray204}{RGB}{204,204,204}
\definecolor{color1}{rgb}{0.0000,0.4470,0.7410}
\definecolor{color2}{rgb}{0.8500,0.3250,0.0980}
\definecolor{color3}{rgb}{0.9290,0.6940,0.1250}
\definecolor{color4}{rgb}{0.4940,0.1840,0.5560}
\definecolor{color5}{rgb}{0.4660,0.6740,0.1880}
\definecolor{color6}{rgb}{0.3010,0.7450,0.9330}
\definecolor{color7}{rgb}{0.6350,0.0780,0.1840}

\begin{axis}[
width=0.48\linewidth,
height=0.42\linewidth,
legend cell align={left},
legend columns = 1,
legend style={
  fill opacity=0.7,
  draw opacity=1,
  text opacity=1,
  at={(1,-0.18)},
  anchor=south east,
  draw=lightgray204,
  font=\scriptsize
},
grid=both,
ticklabel style = {font=\footnotesize},
label style = {font=\footnotesize},
x grid style={darkgray176},
xtick={0.01,0.02,0.03,0.04,0.05,0.06,0.07,0.08,0.09,0.10},
xticklabels={0.01,0.02,0.03,0.04,0.05,0.06,0.07,0.08,0.09,0.10},
xmin=0.01,xmax=0.10,
ymin=1e-4,ymax=1.0,
xlabel={Depolarizing error rate $p$ (d=17)},
xlabel style={
  at={(axis description cs:0.5,-0.3)},
  anchor=north
},
ymode=log,
y grid style={darkgray176},
ylabel={Non-Convergence Rate},
ylabel style={
  at={(axis description cs:0.05,0.40)},
  anchor=south
},
yscale=0.77
]

\addplot[thick, dashed, black, mark=o, 
mark options={draw=black, scale=1, solid}, smooth] 
table[x=p, y expr = \thisrow{errors}/10000] {
p errors
0.01 2775
0.02 7156
0.03 9315
0.04 9896
0.05 9988
0.06 10000
0.07 10000
0.08 10000
0.09 10000
0.1  10000
};
\addlegendentry{No Dilution, $\epsilon=0$}

\addplot[thick, black, mark=*, 
mark options={draw=black, scale=1.0}, smooth] 
table[x=p, y expr = \thisrow{errors}/10000] {
p errors
0.01 1367
0.02 3908
0.03 6137
0.04 7644
0.05 8562
0.06 9131
0.07 9347
0.08 9777
0.09 9897
0.1  9936
};
\addlegendentry{No Dilution, $\epsilon=0.15$}

\addplot[thick, dashed, color7, mark=o, 
mark options={draw=color7, solid, scale=1}, smooth] 
table[x=p, y expr = \thisrow{errors}/10000] {
p errors
0.01 0
0.02 1
0.03 11
0.04 68
0.05 236
0.06 598
0.07 984
0.08 1710
0.09 2510
0.1  3050
};
\addlegendentry{Dilution, $\epsilon=0$}

1367
3908
6137
7644
8562
9131
9347
9777
9897
9936

\addplot[thick, color7, mark=*, 
mark options={draw=color7, scale=1}, smooth] 
table[x=p, y expr = \thisrow{errors}/100000] {
p errors
0.01 0
0.02 0
0.03 0
0.04 0
0.05 11
0.06 13
0.07 15
0.08 21
0.09 20
0.1  21
};
\addlegendentry{Dilution, $\epsilon=0.15$}

0
0
0
4
4
5
3
3
5
4

1
2
22
79
256
600
900
2100
3110
3050

2775
7156
9315
9896
9988

\end{axis}
\end{tikzpicture}
\vspace{-2mm}
\caption{
Left: Average decoding complexity, measured by the median number of iterations
until convergence over \(10^4\) trials. The inset shows the maximum iteration
budget \(I_{\max}\) as a function of code distance. Right: Non-convergence rate
with and without dilution and damping, where \(\epsilon\) denotes the damping factor.
}
\label{fig:complexity}
\end{figure}
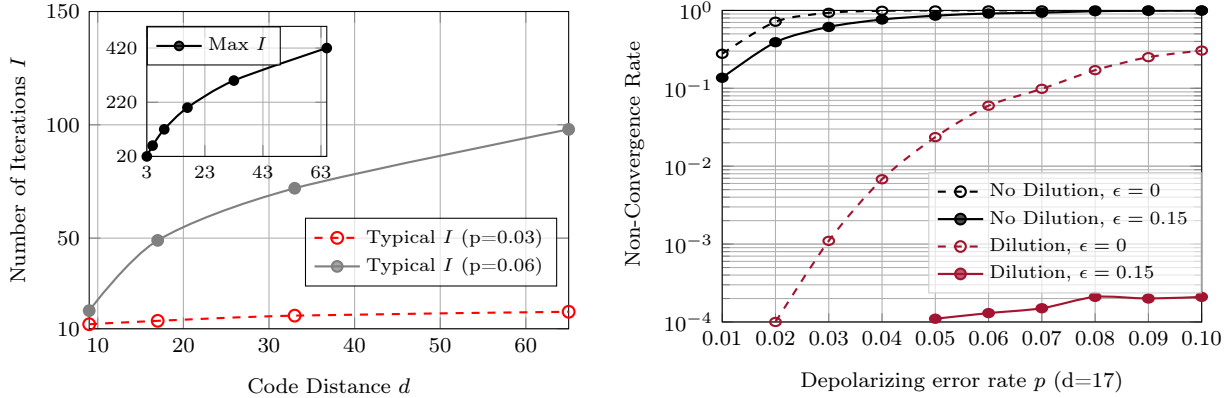
\vspace{-3mm}
\noindent
Next, we examine numerically how dilution affects the convergence behavior of the Min-Sum decoder. As shown in the right panel of Fig.~\ref{fig:complexity}, naive Min-Sum decoding on the original Tanner graph essentially fails to converge. Introducing a damping factor provides only marginal improvement and does not fundamentally resolve the issue. In contrast, combining dilution with damping robustly restores convergence for $d=17$. 
Fig.~\ref{fig:hist_convergence} shows histograms of the number of iterations until convergence at an error rate near threshold for distances $d=9$ (left) and $d=65$ (right). The histograms exhibit multiple peaks, each corresponding to
convergence at a different dilution stage. This indicates that successive
dilution stages unlock additional sets of instances for which Min-Sum
converges. For the smaller code, \(d=9\), naive Min-Sum can still converge, but
typically much more slowly; dilution substantially accelerates convergence.
For the larger code, \(d=65\), the baseline decoder essentially fails. In this
regime, dilution is not merely an acceleration mechanism but appears to be the
decisive ingredient enabling convergence at all. Interestingly, for \(d=65\), most convergent instances are resolved only at the later dilution stages. We discuss the origin of this behavior in Sec.~\ref{sec:analysis}.
\vspace{-0mm}
\begin{figure}[H]
  \centering
  \begin{subfigure}{0.49\textwidth}
    \centering
    \includegraphics[width=\linewidth]{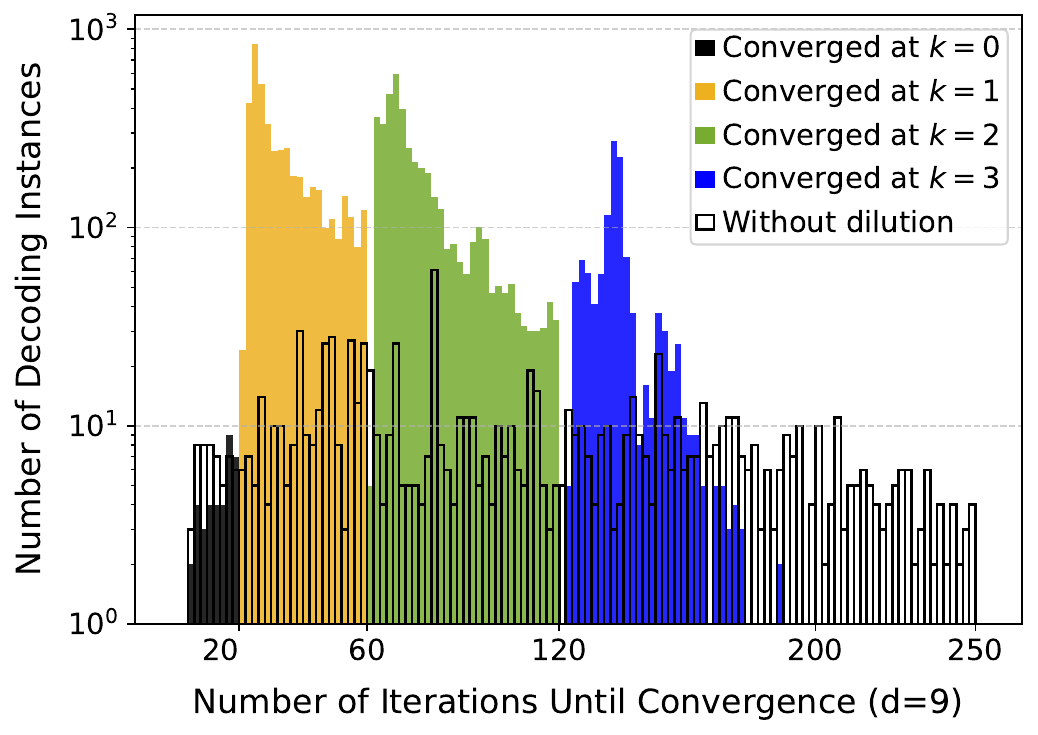}
  \end{subfigure}
  \hfill
  \begin{subfigure}{0.49\textwidth}
    \centering
    \includegraphics[width=\linewidth]{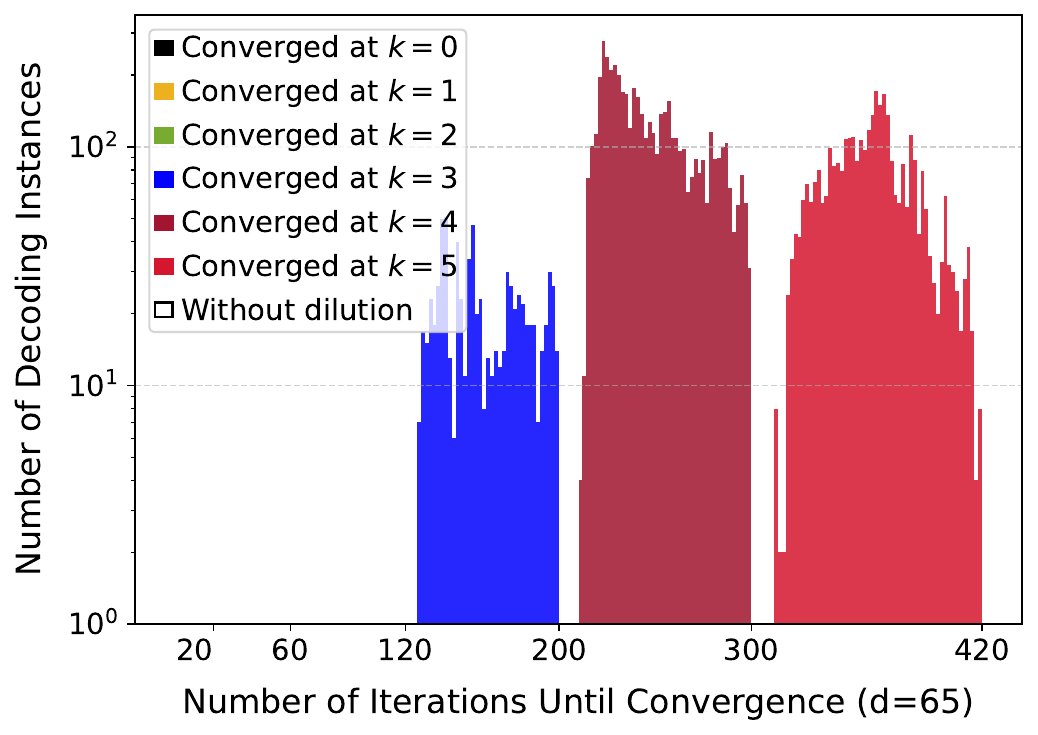}
  \end{subfigure}
  \vspace{-2mm}
  \caption{Histograms of the number of iterations until convergence for the $d=9$ (left) and $d=65$ (right) surface code at threshold $9\%$. The black line shows the baseline without dilution.}
  \label{fig:hist_convergence}
\end{figure}

\section{Interpretation of Dilution Method}
\label{sec:analysis}

While many practical modifications have improved MP decoders for quantum LDPC codes, their behavior remains poorly understood. In this section, we use the dilution method as a lens to first revisit the common view that BP is limited by stabilizer degeneracy and the unavoidable short cycles in Tanner graphs~\cite{poulin2008iterative,babar2015fifteen}. We suggest that this cycle-based perspective is more subtle than it first appears: the correlation structure of the (effective) error distribution also plays a central role. We then use a one-dimensional strip as a toy model to give a heuristic graph-dilution argument for how structure-aware dilution reshapes the effective error distribution across stages, thereby enabling scalable decoding.

\subsection{From Cycle to Correlations: Loops in the Diluted Tanner Graph}
\label{sec: from cycle to correlation}
A major algorithmic obstacle for BP is non-convergence. Conventional wisdom attributes this problem largely to short cycles in the decoding graph. The Tanner graphs of stabilizer codes are known to contain many short loops~\cite{babar2015fifteen}, especially the inevitable mixed-type $4$-cycles formed by different stabilizer types. As shown in Fig.~\ref{fig: diluted tanner} and Fig.~\ref{fig: dilution_sequence}, dilution does not remove all such cycles: each diluted graph still contains many $4$- and $8$-cycles of the types shown on the left of Fig. \ref{fig:qcavity}. From a pure cycle-based viewpoint, one might suspect that the dilution method should not help quaternary BP under depolarizing noise. However, numerical results show that despite the presence of these short cycles, dilution substantially reduces non-convergence up to a distance of $20$. This suggests that the relation between cycles and BP convergence is more subtle than girth alone indicates. In this section, we examine this relation from a more first-principles perspective.

To begin, we recall the implicit independence assumption underlying local message-passing updates. For an \(X\)-type factor node \(a^X \in F\)
(resp. a \(Z\)-type factor node), the factor-to-variable
updates in Eqs.~\ref{eq:f-2-v, bp} and \ref{eq:f-2-v, ms} combine incoming \(Z\)-messages (resp. \(X\)-messages) as if they were independent. From this viewpoint, the effect of cycles is not determined only by their presence, but by the correlations they induce among the incoming messages to a factor node. To make this idea more concrete, consider the short $4$-cycle shown in Fig.~\ref{fig:qcavity}. After two iterations, the incoming messages $\nu_{i \to a^Z}^{(2)}(x_i)$ and $\nu_{j \to a^Z}^{(2)}(x_j)$ correspond to the single-site marginal $\mu^{\backslash a}(x_i \mid \sigma_X)$ and $\mu^{\backslash a}(x_j \mid \sigma_X)$, where $\mu^{\backslash a}(\underline{x} \mid \underline{\sigma})$ is the cavity distribution obtained by removing factor node $a$. Thus, the harmfulness of the cycle can be quantified by the deviation between the joint cavity marginal $\mu^{\backslash a}(\underline{x}_{\partial a^Z} \mid \underline{\sigma})$ and the product of its single-site cavity marginals $\prod_{i \in \partial a^Z} \mu^{\backslash a}(x_i \mid { \sigma_X})$. We denote this deviation by \(\delta(a\mid\underline{\sigma})\) and call it the \emph{cavity discrepancy}. Taking \(\|\cdot\|\) to be total variation distance, define
\begin{equation}
\delta(a^Z \mid \underline{\sigma}) =  \Big\|\, \mu^{\backslash a}(\underline{x}_{\partial a^Z} \mid \underline{\sigma}) - \prod_{i \in \partial a^Z} \mu^{\backslash a}(x_i \mid { \underline{\sigma}}) \,\Big\|\; 
, \; \; \;
\delta(a^X \mid \underline{\sigma}) = \Big\|\, \mu^{\backslash a}(\underline{z}_{\partial a^X} \mid \underline{\sigma}) - \prod_{i \in \partial a^X} \mu^{\backslash a}(z_i \mid { \underline{\sigma}}) \,\Big\|\;.   
\label{eq:uncorrelated}
\end{equation}
\vspace{-3mm}
\begin{figure}[H]
    \centering
    \scalebox{1}{
    \includegraphics[width=1\textwidth]{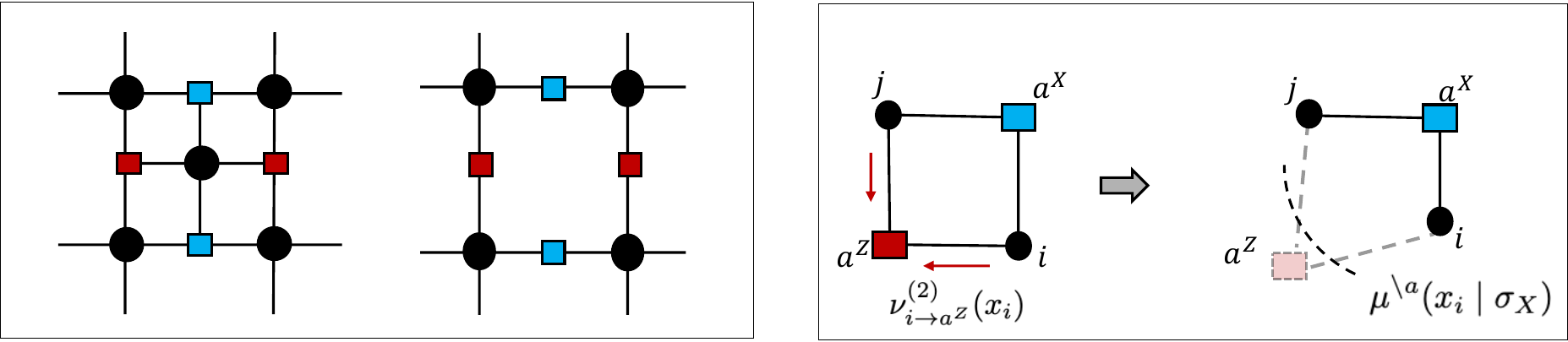} 
    }
    \vspace{-5 mm}
    \caption{Left: Short $4$-cycles and $8$-cycles in the diluted graph formed by different stabilizer types. 
    Right: The cavity picture of the mixed-type $4$-cycle. For this cycle, we are interested in the correlation between the incoming messages \(\nu^{(2)}_{i\to a^Z}(x_i)\) and \(\nu^{(2)}_{j\to a^Z}(x_j)\), shown by the red arrows. After removing the \(Z\)-factor node \(a^Z\), these incoming messages correspond to the cavity marginals \(\mu^{\backslash a}(x_i\mid \sigma_X)\) and \(\mu^{\backslash a}(x_j\mid \sigma_X)\).}
    \vspace{-2 mm}
    \label{fig:qcavity}
\end{figure}
\noindent
To build intuition for the cavity discrepancy, consider a tree factor graph. Removing a factor node disconnects its neighboring variables, so the joint cavity marginal factorizes and \(\delta(a \mid \underline{\sigma})=0\) for every factor node. 
At the opposite extreme, consider a short four-cycle consisting of two variable nodes and two stabilizer factors of the same type. Removing one factor leaves the two variables directly constrained by the other, and hence the cavity discrepancy is essentially maximal. Thus, single-type short cycles are indeed highly harmful.
For mixed-type \(4\)-cycles in Fig.~\ref{fig:qcavity}, however, the cavity discrepancy $\delta(a \mid \underline{\sigma})$ also depends on the coupling between the \(X\)- and \(Z\)-components of the error prior $\psi_i(e_i)=\psi_i(x_i, z_i)$, as illustrated by Proposition~\ref{prop:cd}.
\begin{proposition}
\label{prop:correla}
Consider the mixed-type $4$-cycle consisting of two variable nodes $i$ and $j$ with error prior $\psi_i(x_i, z_i) = \psi_j(x_j, z_j)$, and two factor nodes $a^X$ and $a^Z$. Let 
$
p_i^x := \sum_{z_i}\psi_i(1,z_i), 
\; 
p_i^z := \sum_{x_i}\psi_i(x_i,1),
$
and define the $X$-$Z$ coupling
$
\kappa = \kappa_{\psi_j} =\kappa_{\psi_i}:=\psi_i(1,1)-p_i^x p_i^z.
$  We have
\begin{equation}
\delta( a^Z \mid \sigma_{a^X}=0)
= 
\frac{2(\kappa)^2}{\big((1-p^Z)^2+(p^Z)^2\big)^2}, 
\qquad
\delta( a^Z \mid \sigma_{a^X}=1)
=
\frac{2(\kappa_{\psi_i})^2}{\big(2p^Z(1-p^Z)\big)^2}.
\end{equation}
\noindent
In particular, the correlation vanishes if and only if $\kappa=0$, i.e., when the $X$- and $Z$-components of the single-qubit belief are independent.
\label{prop:cd}
\end{proposition}
\begin{proof}
    A detailed proof is given in Appendix~\ref{ap:2}.
\end{proof}
\noindent
Proposition~\ref{prop:cd} shows that, for a mixed-type $4$-cycle, the cavity discrepancy scales as $\kappa^2$. Under depolarizing noise, with $\psi_{i}(e_i) =  p_{E_i}(e_i) = (1-p, p/3, p/3, p/3)$, one has $\kappa^2 \approx p^2/9$, which leads to small $\delta$. Thus, when the error rate is not too large, the cavity discrepancy induced by each mixed-type \(4\)-cycle is small. This helps explain why, at finite code lengths and with a bounded number of iterations, the correlations among incoming messages remain controlled, and the BP update equations can still provide a good approximation.

The above discussion adopts a correlation-based perspective, explaining why mixed-type $4$-cycles are much more benign than single-type short cycles for BP under depolarizing noise. Under single-$X$ noise, the implication is that graph geometry alone might not be enough to assess BP effectiveness; what matters is the interplay between the (effective) error statistics and the graph geometry.  This means that locally tree-like structure and large girth are not sufficient on their own: we should also ensure that the number of effective errors inside each cycle is much smaller than half the girth.

\subsection{Interpretation of the Dynamics}

In this section, we provide a heuristic explanation for why structure-aware dilution may enable BP to operate at large system sizes and sustain its threshold. To isolate the main mechanism, we consider the simplified one-dimensional strip with periodic boundary conditions shown on the left of Fig.~\ref{fig:1ddilution}. The relevant logical operator is the non-contractible vertical loop of length \(d\). Throughout this toy model, we consider a single $X$-noise with uniform physical error rate $p$. Similar to the graph dilution sequence in Definition~\ref{def:dilution-sequence}, for a strip of length $d$, let $\mathcal{G}^{(0)}, \mathcal{G}^{(1)}, ..., \mathcal{G}^{(K)}$ be the sequence of diluted strips with $K = \lfloor\log_2 d\rfloor$ and $\mathcal{G}^{(k)} = \mathcal{G}^{s_k}$ ($s_k = 2^k-1$). 
For the BP iterations schedule, we choose $I^{\prime}_k = \min \{2^k,\, 20k+20\}$. This gives $I'_k \leq 2^k < g^{(k)}/2 =2^{k}+1$, so the BP computation tree at each stage
\(k\) is cycle-free. Hence, information cannot wrap around a cycle, and the incoming BP messages remain independent. 
It is clear from Fig.~\ref{fig:1ddilution} that moving from stage $k$ to stage $k+1$ increases the girth of the underlying diluted strip from $g^{(k)}$ to $g^{(k+1)}$. However, as discussed in Sec.~\ref{sec: from cycle to correlation}, large girth alone is not sufficient to guarantee BP convergence or decoding performance. We must also control the effective error $\underline{e}^{(k)}$ (see Definition~\ref{def:effective}) at each stage.
Assuming that the initial physical error has weight less than \(d/2\), a natural sufficient condition for both convergence and correctness is that the effective error weight remains below half the code distance throughout the dilution sequence:
\begin{equation}
    n^{(k)} = |\underline e^{(k)}| < \frac d2,
    \qquad \forall k=0, 1, \ldots,K .
\end{equation}
If this condition holds, the final correction is stabilizer-equivalent
to the initial physical error $\underline{e}^{(0)}$. In the following, we give a heuristic argument that the condition is satisfied for all \(d\) whenever the initial error weight obeys $n^{(0)}=|\underline{e}^{(0)}| < d/4$. 
\begin{figure}[H]
    \centering
    \scalebox{1}{
    \includegraphics[width=1\textwidth]{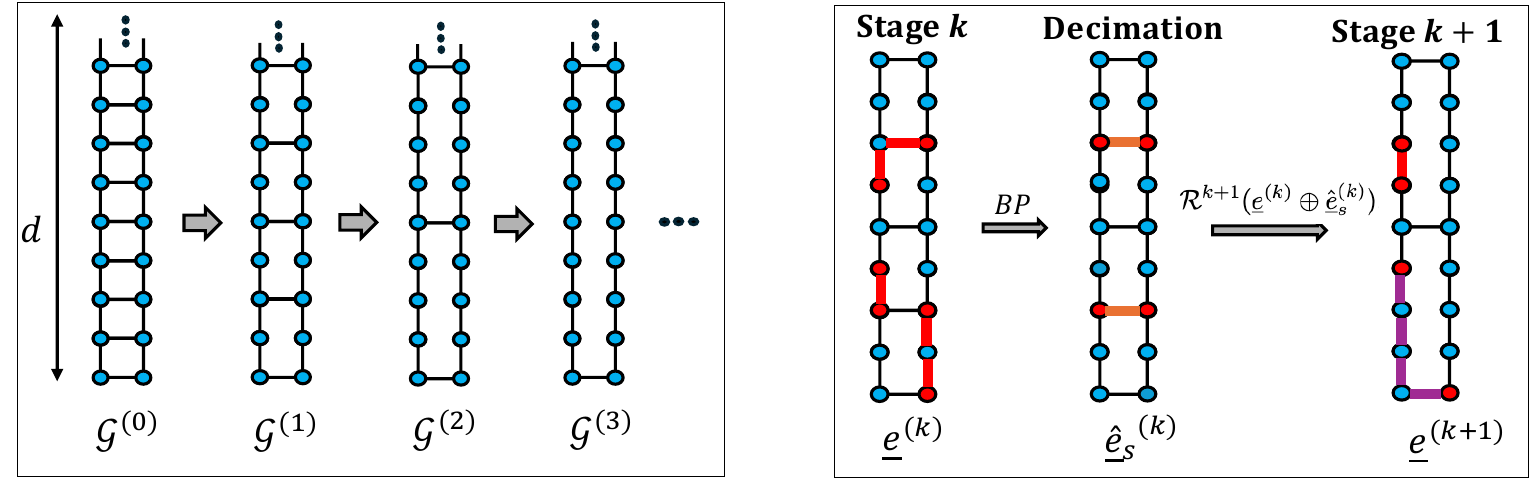} 
    }
    \vspace{-5 mm}
    \caption{Schematic illustration of the dilution process on a one-dimensional strip. Left: The graph-dilution sequence $\mathcal G^{(0)}$, $\mathcal G^{(1)}$, $\mathcal G^{(2)}$, $\ldots$. Right: One step in the evolution of the effective error from stage $k$ to stage $k+1$. BP is first applied on $\mathcal G^{(k)}$ to estimate the errors on the edges that will be removed. After decimation, diluting the corresponding qubits induces a renormalization of the residual error, described by $\mathcal R^{k+1}$, which produces the effective error $\underline{e}^{(k+1)}$ on the diluted strip $\mathcal G^{(k+1)}$.}
    \vspace{-2 mm}
    \label{fig:1ddilution}
\end{figure}
\noindent
To describe how the effective error evolves across dilution stages, start with an effective error $\underline{e}^{(k)}$ on the current strip $\mathcal{G}^{(k)}$. We run BP and decimate the edges removed in the transition from $\mathcal{G}^{(k)}$ to $\mathcal{G}^{(k+1)}$, producing an estimated error $\hat{\underline{e}}_s^{(k)}$ supported on $\mathcal{G}^{(k)} \backslash \mathcal{G}^{(k+1)}$. After decimation, the residual error is $\underline{e}^{(k)} \oplus \hat{\underline{e}}_s^{(k)}$. Since this residue is still defined on $\mathcal{G}^{(k)}$, it is then renormalized to an effective error on the next strip $\mathcal{G}^{(k+1)}$, as shown in the middle of Fig.~\ref{fig:1ddilution}. Thus, BP, decimation, and the global renormalization map $\mathcal{R}$ (see Definition~\ref{def:effective}) together define the stage-$(k+1)$ dilution map, denoted by $\pi^{(k+1)}$:
\begin{equation}
    \underline{e}^{(k+1)} = \pi^{(k+1)}(\underline{e}^{(k)})= \mathcal{R}^{k+1}(\underline{e}^{(k)} \oplus \hat{\underline{e}}_s^{(k)}).
\end{equation}

\noindent
Having formalized the dilution map, we now argue that the effective error weight $n^{(k)}$ does not blow up under successive application of $\pi^{(k)} \circ \cdots \circ \pi^{(1)}$. First, consider the ideal case where every decimation step is correct. Then, at each stage $k$, the effective error $\underline{e}^{(k)}$ is simply the restriction of the initial error $\underline{e}^{(0)}$ to the diluted strip $\mathcal{G}^{(k)}$. Therefore, $n^{(k)} \leq n^{(0)}$. In this best-case scenario, if the initial error satisfies $n^{(0)}<d/2$, then $n^{(k)} < d/2$ for all stages.
In practice, however, wrong decimations are unavoidable, especially at early dilution stages. The key is that wrong decimation does not necessarily cause a large increase in the effective error weight. The intuition is that if the neighborhood of a decimated edge contains only small and well-separated errors, BP is likely to decimate it correctly. Therefore, wrong decimations mainly occur when the local neighborhood already contains large error clusters. In this case, the residue error $\underline{e}^{(k)} \oplus \hat{\underline{e}}^{(k)}$ can often be renormalized, via stabilizer equivalence, into a lower-weight effective error on the next diluted graph, as illustrated by the example in the right of Fig.~\ref{fig:1ddilution}. 

We conduct a numerical experiment to test the above intuition. To describe the experiment, we first introduce \emph{cells} and \emph{blocks}. As shown on the left of Fig.~\ref{fig:numericalexp}, at stage \(k\), the strip \(\mathcal G^{(k)}\) is partitioned into cells
$
\{C_i^{(k)}\}_{i\in \mathcal C_k},
$
where the index set \(\mathcal C_k\) depends on the dilution stage. A block \(B_j^{(k)}\) is defined as the union of two consecutive cells,
$
B_j^{(k)}
=
C_{2j}^{(k)} \cup C_{2j+1}^{(k)},
\; j\in \mathcal B_k.
$
Locally, at each dilution step, each block \(B_j^{(k)}\) is mapped to a cell \(C_j^{(k+1)}\). The experiment is therefore performed on a single stage-\(k\) block \(B_j^{(k)}\).
For a fixed input block weight \(n_B^{(k)}\), we sample many error configurations and compute the weighted average output weight
$
\bar n_C^{(k+1)}(w_{\mathrm{corr}},w_{\mathrm{wrong}})
$
after decimation and renormalization, where \(w_{\mathrm{corr}}\) and \(w_{\mathrm{wrong}}\) are the weights assigned to samples with correct and wrong decimations. The goal is to test whether wrong decimations can make \(\bar n_C^{(k+1)}(w_{\mathrm{corr}},w_{\mathrm{wrong}})\) much larger than the input weight \(n_B^{(k)}\). As shown in the right of Fig.~\ref{fig:numericalexp}, for each \(k\), even under a strong bias toward wrong decimations, \(w_c:w_w=1:100\), the average output weight remains close to the input weight. This supports the intuition that wrong decimations do not cause uncontrolled growth of the effective error weight.
\vspace{-3mm}
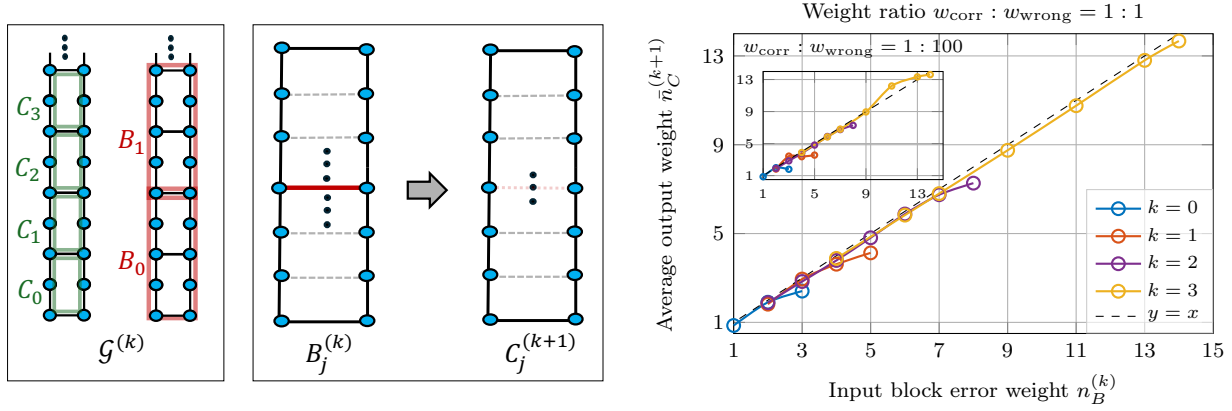
\begin{figure}[H]
\centering
\begin{minipage}{0.48\linewidth}
    \centering
    \includegraphics[width=\linewidth]{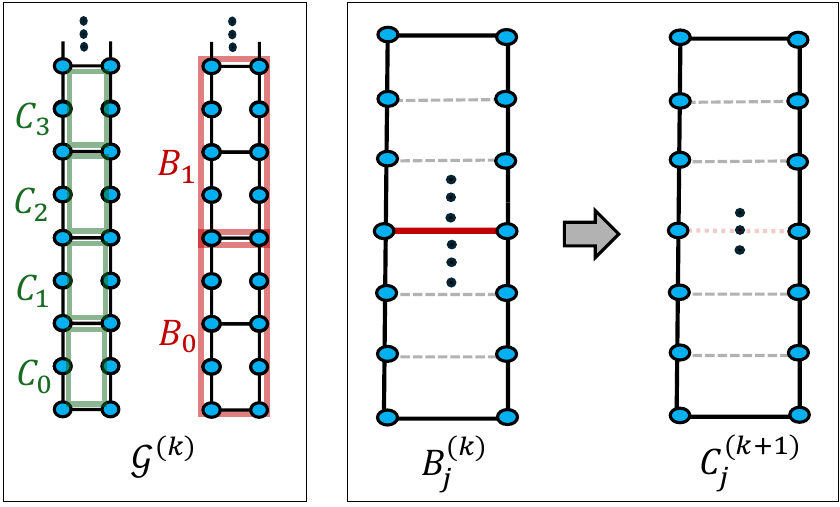}
\end{minipage}
\hspace{2mm}
\begin{minipage}{0.48\linewidth}
    \centering
    \begin{tikzpicture}
\definecolor{darkgray176}{RGB}{176,176,176}
\definecolor{lightgray204}{RGB}{204,204,204}
\definecolor{color1}{rgb}{0.0000,0.4470,0.7410}
\definecolor{color2}{rgb}{0.8500,0.3250,0.0980}
\definecolor{color3}{rgb}{0.9290,0.6940,0.1250}
\definecolor{color4}{rgb}{0.4940,0.1840,0.5560}
\definecolor{color5}{rgb}{0.4660,0.6740,0.1880}
\definecolor{color6}{rgb}{0.3010,0.7450,0.9330}
\definecolor{color7}{rgb}{0.6350,0.0780,0.1840}
\definecolor{color8}{rgb}{0.8350,0.0880,0.1840}

\begin{axis}[
name=mainplot,
title={Weight ratio $w_{\mathrm{corr}}:w_{\mathrm{wrong}} = 1:1$},
title style={font=\footnotesize, yshift=-2mm},
width=\linewidth,
height=0.7\linewidth,
legend cell align={left},
legend columns=1,
legend style={
  fill opacity=0.7,
  draw opacity=1,
  text opacity=1,
  at={(1,0.0)},
  anchor=south east,
  draw=lightgray204,
  font=\scriptsize
},
grid=both,
ticklabel style={font=\footnotesize},
label style={font=\footnotesize},
x grid style={darkgray176},
y grid style={darkgray176},
xmin=1,   xmax=15,
ymin=0.5, ymax=14,
xtick={1, 3, 5, 7, 9, 11, 13, 15},
xlabel={Input block error weight $n_{B}^{(k)}$},
xlabel style={
  at={(axis description cs: 0.5, 0.03)},
  anchor=north
},
ytick={1, 5, 9, 13},
ylabel={Average output weight $\bar{n}_C^{(k+1)}$},
  ylabel style={
    at={(axis description cs: 0.1, 0.5)},
    anchor=south,
  },
]

\addplot[thick, color1, mark=o, mark options={solid, scale=1.2}, smooth]
coordinates {
    (1, 0.8592)
    (2, 1.9063)
    (3, 2.408)
};
\addlegendentry{$k=0$}

\addplot[thick, color2, mark=o, mark options={solid, scale=1.2}, smooth]
coordinates {
    (2, 1.82)
    (3, 2.95)
    (4, 3.62)
    (5, 4.13)
};
\addlegendentry{$k=1$}

\addplot[thick, color4, mark=o, mark options={solid, scale=1.2}, smooth]
coordinates {
    (2, 1.89)
    (3, 2.84)
    (4, 3.789)
    (5, 4.81)
    (6, 5.88)
    (7, 6.748)
    (8, 7.27)
};
\addlegendentry{$k=2$}

\addplot[thick, solid, color3, mark=o, mark options={solid, scale=1.2}, smooth]
coordinates {
    (4,  3.88)
    (6,  5.83)
    (7,  6.80)
    (9,  8.75)
    (11, 10.757)
    (13, 12.80)
    (14, 13.67)
};
\addlegendentry{$k=3$}

y = x
\addplot[
    black,
    dashed,
    domain=0.0:20,
    samples=1000
] {x};
\addlegendentry{$y=x$}

\end{axis}

\begin{axis}[
width=0.5\linewidth,
height=0.38\linewidth,
at={(mainplot.north west)},
anchor=north west,
xshift=0.39cm,
yshift=-0.50cm,
xmin=1, xmax=15,
ymin=0.5, ymax=14,
xtick={1,5,9,13},
ytick={1,5,9,13},
ticklabel style={font=\tiny},
label style={font=\tiny},
title={$w_{\mathrm{corr}}:w_{\mathrm{wrong}} = 1:100$},
title style={font=\scriptsize, yshift=-1.5mm},
grid=both,
x grid style={darkgray176},
y grid style={darkgray176},
legend style={draw=none, fill=none},
]
\addplot[thick, color1, mark=o, mark options={solid, scale=0.5}, smooth]
coordinates {
    (1, 0.8592)
    (2, 1.98)
    (3, 1.79)
};

\addplot[thick, color2, mark=o, mark options={solid, scale=0.5}, smooth]
coordinates {
    (2, 1.82)
    (3, 3.43)
    (4, 3.36)
    (5, 3.56)
};

\addplot[thick, color4, mark=o, mark options={solid, scale=0.5}, smooth]
coordinates {
    (2, 1.89)
    (3, 2.84)
    (4, 3.789)
    (5, 4.81)
    (6, 5.88)
    (7, 6.748)
    (8, 7.27)
};

\addplot[thick, solid, color3, mark=o, mark options={solid, scale=0.5}, smooth]
coordinates {
    (4,  3.88)
    (6,  5.83)
    (7,  6.80)
    (9,  8.95)
    (11, 12.19)
    (13, 13.33)
    (14, 13.61)
};

y = x
\addplot[
    black,
    dashed,
    domain=0.0:20,
    samples=1000
] {x};
\end{axis}

\end{tikzpicture}
\end{minipage}
\vspace{-2mm}
\caption{Numerical experiment for testing the evolution of effective error weight under dilution. The experiment is performed locally on one block $B^{(k)}$ at different stages $k$. Left: Illustration of cells and blocks, green regions denote cells, while red regions denote blocks. Each block consists of two neighboring cells. Middle: A stage-$k$ block $B_j^{(k)}$, formed by two neighboring cells, is mapped to a stage-$(k+1)$ cell $C_j^{(k+1)}$. The red edge denotes the qubit that is decimated when passing to stage $k+1$. Right: Weighted average output error weight $\bar n_C^{(k+1)}$ as a function of the input block error weight $n_B^{(k)}$. Each color represents a different dilution stage $k$. The main panel uses equal weights, $w_{\mathrm{corr}}:w_{\mathrm{wrong}}=1:1$, while the inset uses $1:100$, strongly biasing the average toward wrong decimations. The dashed line is $y=x$ as a reference.}
\label{fig:numericalexp}
\end{figure}
\vspace{-3mm}
\noindent
\noindent
The numerical experiment suggests that the block-to-cell renormalization is
non-expansive at the level of local effective weight. Although this experiment
is performed locally on a single block, we use it as evidence for the following heuristic assumption on the true dilution dynamics: for sufficiently large \(d\),
\begin{equation}
    \sum_{i\in \mathcal C_{k+1}}
    n^{(k+1)}_{C_i}
    \leq
    \sum_{j\in \mathcal B_k}
    n^{(k)}_{B_j},
    \label{eq:heuristic-local-nonincreasing}
\end{equation}
where $n^{(k)}_{B_j}=|\underline e^{(k)}|_{B_j}$ and $n^{(k)}_{C_i}=|\underline e^{(k)}|_{C_i}$ denote the effective error weight restricted to block \(B_j^{(k)}\) and cell \(C_i^{(k)}\), respectively.
Since the Hamming weight is bounded by any covering count, and by the covering relation between blocks and cells at the same stage, we have
$
    n^{(k)}
    =
    |\underline e^{(k)}|
    \leq
    \sum_{j\in\mathcal B_k} n^{(k)}_{B_j}
    \leq 
    \sum_{i\in\mathcal C_k} n^{(k)}_{C_i}.
$
Using the heuristic non-expansion assumption in
Eq.~\eqref{eq:heuristic-local-nonincreasing} recursively gives
\begin{equation}
\begin{aligned}
    n^{(k)}
    \leq
    \sum_{i\in\mathcal C_k} n^{(k)}_{C_i}
    \leq
    \sum_{j\in\mathcal B_{k-1}} n^{(k-1)}_{B_j}
    \leq
    \sum_{i\in\mathcal C_{k-1}} n^{(k-1)}_{C_i}
    \leq
    \cdots
    \leq
    \sum_{j=1}^{d/2} n^{(0)}_{B_j}.
\end{aligned}
\label{eq:heuristic-weight-bound}
\end{equation}
Thus, the stage-wise condition is satisfied if the initial block-counting
weight obeys
$
    \sum_{j=1}^{d/2} n^{(0)}_{B_j} < \frac d2 .
$
Since each physical error is counted in at most two initial blocks, we have
$
    \sum_{j=1}^{d/2} n^{(0)}_{B_j}
    \leq
    2|\underline e^{(0)}|.
$
Therefore, a sufficient condition in terms of the initial physical error
weight is
$
    |\underline e^{(0)}| < \frac d4 .
$

\section{Discussion and Conclusion}
In summary, this work takes a step toward algorithmically scalable and interpretable MP-based decoding approaches for quantum error correction, seeking to reconcile the tension between practice and theory.
Looking forward, a direction is to generalize the dilution method for broader families of quantum LDPC codes under a more realistic noise model. Ultimately, we hope to find a framework that enables message-passing decoding to work reliably for quantum LDPC codes, meaning that the method is robust and scalable, and its behavior can largely be understood and predicted by the message-passing theory.

\section*{Acknowledgments}
The authors would like to thank Di Wu and Eduardo Mucciolo for their insightful discussions.


\newpage
\bibliographystyle{IEEEtran}
\bibliography{refs}

\newpage
\appendix
\section{Proof of Theorem~\ref{tm:distance}}
\label{ap:1}

We use two different local decompositions for the two sparsification patterns, as shown in the Fig.~\ref{fig:CD}. For the diagonal pattern \(D_H^s\), \(D_i\) denotes a local vertical strip; its interior contains the vertical edges of the strip, and \(\partial D_i\) denotes the retained horizontal edges on its boundary.  These regions form an overlapping cover. For the Cartesian pattern \(C_H^s\), \(C_i\) denotes a
horizontal band extending across the lattice; its interior contains the
vertical edges in the band, and \(\partial C_i\) consists of the two retained
horizontal boundary lines. These regions give a disjoint decomposition of the
diluted lattice.

\begin{figure}[H]
    \centering
    \scalebox{1}{
    \includegraphics[width=1\textwidth]{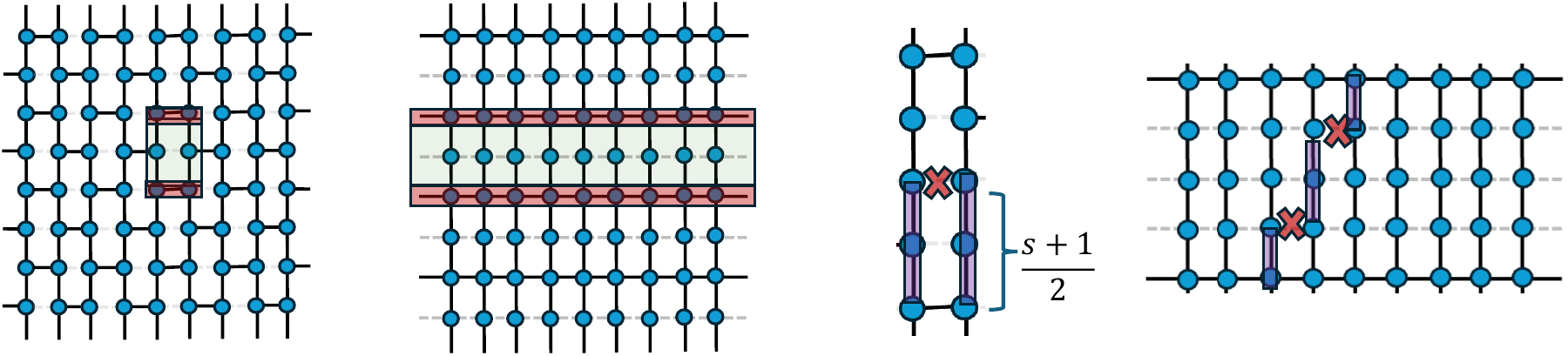} 
    }
    \vspace{-5 mm}
    \caption{Local regions and worst-case configurations used in Lemma~\ref{lm:interiorboundary}.
The first two panels define the local decompositions. For the diagonal pattern,
\(D_i\) is the green local region and \(\partial D_i\) consists of the red
retained boundary edges. For the Cartesian pattern, \(C_i\) is the green
horizontal band and \(\partial C_i\) consists of the two red retained boundary
lines. The last two panels illustrate worst-case local configurations. In the
diagonal region, an error placed midway between retained boundaries produces
the largest local expansion, giving the factor \(2\lfloor (s+1)/2\rfloor\). In the Cartesian band, the worst-case imbalance between the effective correction and the competing logical representative is achieved by placing the errors midway between the two retained boundary lines, giving the factor \(\lfloor (s-1)/2\rfloor\).}
    \vspace{-2 mm}
    \label{fig:CD}
\end{figure}

\begin{lemma}[Local expansion bounds]
\label{lm:interiorboundary}
Fix a sparsification ratio \(s\geq 1\). For the diagonal pattern
\(D_H^s\), let \(D_i\) denote the local vertical region shown in
Fig.~\ref{fig:CD}: its interior consists of the vertical edges
inside the region, while \(\partial D_i\) consists of the retained horizontal
boundary edges. For the Cartesian pattern \(C_H^s\), let \(C_i\) denote the
horizontal band shown in Fig.~\ref{fig:CD}: its interior consists
of the vertical edges in the band, while \(\partial C_i\) consists of the two
retained horizontal boundary lines.
Let \(\underline e\) be a nonzero error configuration restricted to one such
region, and let \(\tilde{\underline e}\) be the corresponding minimum-weight effective error supported on the diluted lattice. Let \(\tilde{\underline \ell}(\underline e)\) denote the minimum-weight effective representative in the opposite logical coset. Then the following local bounds hold:
\begin{equation}
    \max_{i,\underline e}
    \frac{
        \big|\tilde{\underline e}\big|_{D_i}
    }{
        \big|\underline e\big|_{D_i}
    }
    =
    \max_{\underline e}
    \frac{
        \big|\tilde{\underline e}\big|_{D}
    }{
        \big|\underline e\big|_{D}
    }
    =
    2\left\lfloor \frac{s+1}{2}\right\rfloor ,
    \label{eq:local-D-bound}
\end{equation}
and
\begin{equation}
    \max_{i,\underline e}
    \frac{
        \big|\tilde{\underline e}\big|_{C_i}
        -
        \big|\tilde{\underline \ell}(\underline e)\big|_{C_i}
    }{
        \big|\underline e\big|_{C_i}
    }
    =
    \max_{\underline e}
    \frac{
        \big|\tilde{\underline e}\big|_{C}
        -
        \big|\tilde{\underline \ell}(\underline e)\big|_{C}
    }{
        \big|\underline e\big|_{C}
    }
    =
    \left\lfloor \frac{s-1}{2}\right\rfloor .
    \label{eq:local-C-bound}
\end{equation}
Here, the maximization is over all nonzero local error configurations, so the
denominators are nonzero.
\end{lemma}

\begin{proof}
The claim follows from the local geometries shown in
Fig.~\ref{fig:CD}. For the diagonal region \(D_i\), the worst case
is obtained by placing a local error midway between two retained boundary
edges. Its effective representative can extend to the nearest retained
boundaries on both sides, giving at most \(2\lfloor (s+1)/2\rfloor\) effective edges per local error. By subadditivity of Hamming weight, the same bound holds for any local error configuration.
For the Cartesian region \(C_i\), the interior contribution compares the
effective representative with the competing representative in the opposite
logical coset. The largest possible difference occurs when the local errors
are placed midway between the two retained boundary lines, as illustrated in
Fig.~\ref{fig:CD}. In this case, the correction returning to the
same boundary line and the competing path crossing to the other boundary line
differ by at most \(\lfloor (s-1)/2\rfloor\) per local error. Subadditivity
again extends the bound to arbitrary local configurations.
\end{proof}

\begin{proof}
    Let $\underline{e}$ be an error configuration on original lattice $\mathcal{G}$ and  $\tilde{\underline{e}}$ be its effective error (see Definition~\ref{def:effective}) supported on diluted lattice $\mathcal{G}^s$, and let $\tilde{\underline{\ell}}(\underline{e})$ denote the minimum-weight effective operator that induces a logical failure (i.e., the minimum-weight representative in the corresponding logical coset). For any error $\underline{e}$ on the original lattice $\mathcal{G}$ with $|\underline{e}| \leq \lfloor (d-1)/2 \rfloor$, it is correctable on a given diluted lattice $\mathcal{G}^s$ if
    \begin{equation}
          \Delta_{\mathcal{G}^s}(\underline{e}) = |\tilde{\underline{e}}|- |\tilde{\underline{\ell}}(\underline{e}) | <0.
          \label{eq:c1}
    \end{equation}
     Use the fact that the Hamming weight is subadditive under a cover and that $\bigcup_i \partial D_i \subset \mathcal{G}_{D_H}^s$. For $s$-diluted lattice via Diagonal dilution $\mathcal{G}^s_{D_H}$,  we obtain
    \begin{equation}
        \big|\tilde{\underline{e}}\big|
        =
        \frac{1}{2} \cdot \sum_i \big|\tilde{\underline{e}}\big|_{D_i}
        +
        \sum_i \big|\tilde{\underline{e}}\big|_{\partial D_i},
        \quad
        \big|\tilde{\underline{\ell}}(\underline{e})\big| 
        \geq
        \sum_i \big|\tilde{\underline{\ell}}(\underline{e})\big|_{\partial D_i}
        \label{eq:subadd-cover}
    \end{equation}
    From \eqref{eq:subadd-cover}, it follows that for the $D_H^s$-diluted lattice $\mathcal{G}^s_{D_H}$,
    \begin{equation}
        \Delta_{\mathcal{G}^s_{D_H}}(\underline{e}) = 
        \frac{1}{2}\cdot \sum_i \big|\tilde{\underline{e}}\big|_{D_i}
        +
        \sum_i \big|\tilde{\underline{e}}\big|_{\partial D_i} 
        -
         \sum_i \big|\tilde{\underline{\ell}}(\underline{e})\big|_{\partial D_i}
         =
         \underbrace{\frac{1}{2}\cdot \sum_i \big|\tilde{\underline{e}}\big|_{D_i}}_{\Delta|_{D_H}} + \underbrace{\sum_i \biggr( \big|\tilde{\underline{e}}\big|_{\partial D_i} 
        -
          \big|\tilde{\underline{\ell}}(\underline{e})\big|_{\partial D_i} \biggr)}_{\Delta|_{\partial D_H}}
    \end{equation}
    \noindent
    For an $C_H^s$-diluted lattice $\mathcal{G}_{C_H}^s$, since $\{C_i\}_i$ and $\{\partial C_i\}_i$ together form a partition of the lattice (i.e., they cover the lattice and are pairwise disjoint), the Hamming weight decomposes exactly as
    \begin{equation}
        \big|\tilde{\underline{e}}\big|
        =
        \sum_i \big|\tilde{\underline{e}}\big|_{C_i}
        +
        \sum_i \big|\tilde{\underline{e}}\big|_{\partial C_i},
        \qquad
        \big|\tilde{\underline{\ell}}(\underline{e})\big|
        =
        \sum_i \big|\tilde{\underline{\ell}}(\underline{e})\big|_{C_i}
        +
        \sum_i \big|\tilde{\underline{\ell}}(\underline{e})\big|_{\partial C_i},
    \end{equation}
    Hence, for $\mathcal{G}_{C_H}^s$, we have
    \begin{equation}
        \Delta_{\mathcal{G}^s_{C_H}}(\underline{e}) =  \sum_i \big|\tilde{\underline{e}}\big|_{C_i} 
    - 
    \sum_i \big|\tilde{\underline{\ell}}(\underline{e})\big|_{C_i}
    + 
    \sum_i \big|\tilde{\underline{e}}\big|_{\partial C_i}
    - 
    \sum_i \big|\tilde{\underline{\ell}}(\underline{e})\big|_{\partial C_i}
    =
    \underbrace{\sum_{i} \biggr( \big|\tilde{\underline{e}}\big|_{C_i} - \big|\tilde{\underline{\ell}}(\underline{e})\big|_{C_i} \biggr) }_{\Delta|_{C_H}}
    +
    \underbrace{\sum_i \biggr(\big|\tilde{\underline{e}}\big|_{\partial C_i}
    - 
    \big|\tilde{\underline{\ell}}(\underline{e})\big|_{\partial C_i} \biggr)}_{\Delta|_{\partial C_H}}
    \end{equation}
     We refer to the first term $\Delta|_{D}, \Delta|_{C}$ as the \emph{interior term} and the second term $\Delta|_{\partial D}, \Delta|_{\partial C}$ as the \emph{boundary term}. 
    \noindent
    For the interior term $\Delta|_{D}$ and $\Delta|_{C}$, we have
    \begin{equation}
        \Delta|_{D_H} = \frac{1}{2}\cdot \sum_i \big|\tilde{\underline{e}}\big|_{D_i} 
        =
        \frac{1}{2} \cdot \sum_{i} \biggr( \frac{\big|\tilde{\underline{e}}\big|_{D_i}}
        { \big|\underline{e} \big|_{D_i}} \cdot \big|\underline{e} \big|_{D_i}\biggr)
        \leq 
         \biggr(\max_{i, \, \underline{e}}  \frac{\big|\tilde{\underline{e}}\big|_{D_i} }
    { \big|\underline{e} \big|_{D_i}} \biggr) \cdot \biggr(\frac{1}{2} \cdot \sum_{i} \big|\underline{e} \big|_{D_i}\biggr)
    \label{eq:interior_D}
    \end{equation}
    Similarly,
    \begin{equation}
        \Delta|_{C_H}=\sum_{i} \biggr( \big|\tilde{\underline{e}}\big|_{C_i} - \big|\tilde{\underline{\ell}}(\underline{e})\big|_{C_i} \biggr) 
        = 
        \sum_{i} \biggr( \frac{\big|\tilde{\underline{e}}\big|_{C_i} - \big|\tilde{\underline{\ell}}(\underline{e}) \big|_{C_i}}
        { \big|\underline{e} \big|_{C_i}} \cdot \big|\underline{e} \big|_{C_i}\biggr)
        \leq 
        \biggr(\max_{i, \, \underline{e}}  \frac{\big|\tilde{\underline{e}}\big|_{C_i} - \big|\tilde{\underline{\ell}}(\underline{e}) \big|_{C_i}}
        { \big|\underline{e} \big|_{C_i}}\biggr) 
        \cdot 
        \biggr( \sum_{i} \big|\underline{e} \big|_{C_i}\biggr)
        \label{eq:interior_C}
    \end{equation}
    Since all blocks $D_i$ (resp. $C_i$) are identical by construction, we have the following inequality by Lemma~\ref{lm:interiorboundary},
    \begin{equation}
        \max_{i, \, \underline{e}}  \frac{\big|\tilde{\underline{e}}\big|_{D_i} }
    { \big|\underline{e} \big|_{D_i}} = \max \frac{\big|\tilde{\underline{e}}\big|_{D_i} }
    { \big|\underline{e} \big|_{D_i}} = 2\cdot \biggr\lfloor \frac{s+1}{2} \biggr\rfloor.
    \end{equation}
    Similarly,
    \begin{equation}
        \max_{C_i, \, \underline{e}}  \frac{\Big|\tilde{\underline{e}}|_{C_i}\Big| - \Big|\tilde{\underline{\ell}}|_{C_i}(\underline{e}) \Big|}
    { \Big|\underline{e}|_{C_i} \Big|}
    = 
    \max  \frac{\Big|\tilde{\underline{e}}|_{C}\Big| - \Big|\tilde{\underline{\ell}}|_{C}(\underline{e}) \Big|}
    { \Big|\underline{e}|_{C} \Big|} = \biggr\lfloor \frac{s-1}{2} \biggr\rfloor. 
    \end{equation}
    Since $\frac{1}{2} \sum_i |\underline e|_{D_i} \le |\underline e|$ and $\sum_i |\underline e|_{C_i} \le |\underline e|$, the interior term are bounded by
    \begin{equation}
        \Delta|_{D_H}
        \le
         2\cdot \Big\lfloor \frac{s+1}{2} \Big\rfloor\,\cdot |\underline e|,
        \qquad
        \Delta|_{C_H}
        \le
        \Big\lfloor \frac{s-1}{2} \Big\rfloor\,\cdot |\underline e|.
    \end{equation}
    \noindent
    For the boundary term $\Delta|_{\partial D_H}$ and $\Delta|_{\partial C_H}$, since
    \begin{equation}
        \sum_{i} \big|\tilde{\underline{e}}\big|_{\partial D_i} + \sum_{i} \big|\tilde{\underline{\ell}}(\,\underline{e}\,) \big|_{\partial D_i} 
        \geq  
        \sum_{i} \big|\tilde{\underline{e}}|_{\partial D_i} \oplus \tilde{\underline{\ell}}|_{\partial D_i}(\,\underline{e}\,) \big| 
        \geq
        \sum_{i} \big|(\tilde{\underline{e}} \oplus \tilde{\underline{\ell}})|_{\partial D_i}(\,\underline{e}\,) \big| 
        \geq 
        d. 
    \end{equation}
    Similarly,
    \begin{equation}
        \sum_{i} \big|\tilde{\underline{e}}\big|_{\partial C_i} + \sum_{i} \big|\tilde{\underline{\ell}}(\,\underline{e}\,) \big|_{\partial C_i} 
        \geq  
        \sum_{i} \big|\tilde{\underline{e}}|_{\partial C_i} \oplus \tilde{\underline{\ell}}|_{\partial C_i}(\,\underline{e}\,) \big| 
        \geq
        \sum_{i} \big|(\tilde{\underline{e}} \oplus \tilde{\underline{\ell}})|_{\partial C_i}(\,\underline{e}\,) \big| 
        \geq 
        d. 
    \end{equation}
    Since $\sum_i \big|\tilde{\underline{e}}\big|_{\partial D_i} \leq  \big|\underline{e}\big|$ and $\sum_i \big|\tilde{\underline{e}}\big|_{\partial C_i} \leq  \big|\underline{e}\big|$, we have
    $$
    \Delta|_{\partial C_H} = \sum_i \biggr(\big|\tilde{\underline{e}}\big|_{\partial C_i}
    - 
    \big|\tilde{\underline{\ell}}(\underline{e})\big|_{\partial C_i} \biggr)
    \leq 
    2 \big|\underline{e} \big| -d,\quad
    \Delta|_{\partial D_H} = \sum_i \biggr(\big|\tilde{\underline{e}}\big|_{\partial D_i}
    - 
    \big|\tilde{\underline{\ell}}(\underline{e})\big|_{\partial D_i} \biggr)
    \leq 
    2 \big|\underline{e} \big| -d
    $$
    Therefore,
    $$
    \Delta(\mathcal{G}^s_{D_H}) = \Delta|_{D_H} + \Delta|_{\partial D_H}  \leq 2 \cdot  \biggr\lfloor \frac{s+1}{2} \biggr\rfloor \cdot \big|\underline{e} \big| + 2\big|\underline{e}\big| -d,
    \quad
    \Delta(\mathcal{G}^s_{C_H}) = \Delta|_{C_H} + \Delta|_{\partial C_H}  \leq \biggr\lfloor \frac{s-1}{2} \biggr\rfloor \cdot \big|\underline{e} \big| + 2\big|\underline{e}\big| -d.
    $$
    $$
    t(\mathcal{G}_D) \geq \Big\lfloor\frac{d-1}{2\cdot \lfloor\frac{s+1}{2}\rfloor + 2} \Big \rfloor, \quad 
    t(\mathcal{G}_{C_H}) \geq \biggr\lfloor\frac{d-1}{\lfloor\frac{s-1}{2}\rfloor + 2} \biggr \rfloor.
    $$
\end{proof}

\section{Proof of Proposition~\ref{prop:correla}}
\label{ap:2}
\begin{proof}
Let \((X,Z)\) denote the pair of random variables representing the \(X\)- and \(Z\)-error components of a single-qubit error. For the two variable nodes \(i\) and \(j\) in the mixed-type \(4\)-cycle, we write
\[
    (X_i,Z_i) \equiv (x_i,z_i),
    \qquad
    (X_j,Z_j) \equiv (x_j,z_j),
\]
for their associated random variables. Since the two qubits have the same single-qubit prior $\psi_i(x_i,z_i)=\psi_j(x_j, z_j) = p_{E_i}(x_i, z_i)$, we denote
\[
    p_{10}:=\psi_i(1,0),
    \qquad
    p_{11}:=\psi_i(1,1),
    \qquad
    p^Z:=\sum_{x\in\{0,1\}}\psi_i(x,1).
\]
After removing the \(Z\)-factor node \(a^Z\), the remaining \(X\)-factor imposes the constraint $Z_i\oplus Z_j=\sigma_{a^X}$. Thus the cavity distribution of \((X_i,X_j)\) is obtained by conditioning on either \(Z_i=Z_j\) or \(Z_i\neq Z_j\). First consider \(\sigma_{a^X}=0\), i.e. \(Z_i=Z_j\). We have
    \begin{equation}
        \mathbb{P}(X_i=1, X_j=1, Z_i=Z_j) 
        = 
        \sum_{i \in {(0,1)}} \mathbb{P}(X_1=1, Z_1=i) \cdot \mathbb{P}(X_2=1, Z_2=i) 
        =
        p_{10}^2 +p_{11}^2,
    \end{equation}
and
\begin{equation}
    \mathbb P(Z_i=Z_j)
    =
    (1-p^Z)^2+(p^Z)^2.
\end{equation}
Hence
\begin{equation}
    \mathbb P(X_i=1,X_j=1\mid Z_i=Z_j)
    =
    \frac{p_{10}^2+p_{11}^2}
    {(1-p^Z)^2+(p^Z)^2}.
\end{equation}
Similarly,
\begin{equation}
    \mathbb P(X_i=1,Z_i=Z_j)
    =
    p_{10}(1-p^Z)+p_{11}\cdot p^Z,
\end{equation}
and therefore
\begin{equation}
    \mathbb P(X_1=i\mid Z_i=Z_j)
    =
    \frac{p_{10}(1-p^Z)+p_{11}p^Z}
    {(1-p^Z)^2+(p^Z)^2}.
\end{equation}
By symmetry, this is also
\(\mathbb P(X_j=1\mid Z_i=Z_j)\). Thus
\begin{align}
    \operatorname{Cov}(X_i,X_j\mid \sigma_{a^X}=0)
    &=
    \mathbb P(X_i=1,X_j=1\mid Z_i=Z_j)
    -
    \mathbb P(X_1=1\mid Z_1=Z_2)^2 \notag \\
    &=
    \frac{p_{10}^2(p^Z)^2+p_{11}^2(1-p^Z)^2
    -2p_{10}p_{11}p^Z(1-p^Z)}
    {\big((1-p^Z)^2+(p^Z)^2\big)^2} \notag \\
    &=
    \frac{\big(p_{11}(1-p^Z)-p_{10}p^Z\big)^2}
    {\big((1-p^Z)^2+(p^Z)^2\big)^2}.
    \label{eq:cov-even}
\end{align}
Next consider \(\sigma_{a^X}=1\), i.e. \(Z_i\neq Z_j\). We have
\begin{equation}
    \mathbb P(X_i=1,X_j=1,Z_i\neq Z_j)
    =
    2p_{10}p_{11},
\end{equation}
and
\begin{equation}
    \mathbb P(Z_i\neq Z_j)
    =
    2p^Z(1-p^Z).
\end{equation}
Therefore
\begin{equation}
    \mathbb P(X_1=i,X_2=j\mid Z_i\neq Z_j)
    =
    \frac{p_{10}p_{11}}
    {p^Z(1-p^Z)}.
\end{equation}
Moreover,
\begin{equation}
    \mathbb P(X_i=1,Z_i\neq Z_j)
    =
    p_{10}p^Z+p_{11}(1-p^Z),
\end{equation}
so
\begin{equation}
    \mathbb P(X_i=1\mid Z_i\neq Z_j)
    =
    \frac{p_{10}p^Z+p_{11}(1-p^Z)}
    {2p^Z(1-p^Z)}.
\end{equation}
Again by symmetry, the same expression holds for \(X_2\). Hence
\begin{align}
    \operatorname{Cov}(X_i,X_j\mid \sigma_{a^X}=1)
    &=
    \mathbb P(X_i=1,X_j=1\mid Z_i\neq Z_j)
    -
    \mathbb P(X_i=1\mid Z_i\neq Z_j)^2 \notag \\
    &=
    -\frac{\big(p_{11}(1-p^Z)-p_{10}p^Z\big)^2}
    {\big(2p^Z(1-p^Z)\big)^2}.
    \label{eq:cov-odd}
\end{align}
\noindent
Finally, by definition of the $X$-$Z$ coupling $\kappa:=\psi_i(1,1)-\sum_{z_i}\psi_i(1,z_i) \cdot \sum_{x_i}\psi_i(x_i,1).$,
\begin{align}
    \kappa
    &=
    \operatorname{Cov}(X,Z) \notag \\
    &=
    \mathbb P(X=1,Z=1)-\mathbb P(X=1)\mathbb P(Z=1) \notag \\
    &=
    p_{11}-(p_{10}+p_{11})p^Z \notag
    =
    p_{11}(1-p^Z)-p_{10}p^Z.
\end{align}
Substituting this into Eqs.~\eqref{eq:cov-even} and \eqref{eq:cov-odd} gives
\begin{equation}
    \operatorname{Cov}(X_i,X_j\mid \sigma_{a^X}=0)
    =
    \frac{\kappa^2}
    {\big((1-p^Z)^2+(p^Z)^2\big)^2},
\end{equation}
and
\begin{equation}
    \operatorname{Cov}(X_i,X_j\mid \sigma_{a^X}=1)
    =
    -\frac{\kappa^2}
    {\big(2p^Z(1-p^Z)\big)^2}.
\end{equation}
Since the cavity discrepancy $\delta(a \mid \underline{\sigma})$ is defined using total variation distance, then for two binary variables with identical marginals,
\[
    \left\|
    \mu(x_1,x_2)-\mu(x_1)\mu(x_2)
    \right\|_{\mathrm{TV}}
    =
    2\left|\operatorname{Cov}(X_1,X_2)\right|.
\]
Therefore,
\begin{equation}
    \delta(a^Z\mid \sigma_{a^X}=0)
    =
    \frac{2\kappa^2}
    {\big((1-p^Z)^2+(p^Z)^2\big)^2},
\end{equation}
and
\begin{equation}
    \delta(a^Z\mid \sigma_{a^X}=1)
    =
    \frac{2\kappa^2}
    {\big(2p^Z(1-p^Z)\big)^2}.
\end{equation}
In particular, the cavity discrepancy vanishes if and only if
\(\kappa=0\), equivalently when the \(X\)- and \(Z\)-components of the single-qubit error prior are independent.
\end{proof}

\newpage
\end{document}